\documentclass[12pt]{article}
\usepackage[letterpaper,margin=1in]{geometry}
\usepackage{setspace}
\usepackage{amsmath,amsfonts,amssymb,amsthm}
\usepackage{bm}
\usepackage{natbib}
\usepackage{changepage}
\usepackage{graphicx}
\usepackage{caption}
\usepackage{subcaption}
\usepackage{rotating}
\usepackage{pdflscape}
\usepackage[utf8]{inputenc}
\usepackage[english]{babel}

\usepackage{booktabs}
\usepackage{multirow}
\usepackage{makecell}
\usepackage{color}
\usepackage{placeins}
\usepackage{float}
\usepackage{mathabx}
\usepackage{afterpage}
\usepackage[export]{adjustbox}
\usepackage{url}
\usepackage{authblk}
\usepackage[colorlinks=true, linkcolor=blue, citecolor=blue, urlcolor=blue]{hyperref}

\usepackage{upgreek}
\usepackage{algorithm}
\usepackage{algpseudocode}
\usepackage{xcolor}

\def\u{{\bf u}}
\newcommand{\bsm}{\mbox{\bf m}}
\newcommand{\bx}{\mbox{\bf x}}
\newcommand{\bM}{\mbox{\bf M}}
\newcommand{\bA}{\mbox{\bf A}}
\newcommand{\bB}{\mbox{\bf B}}
\newcommand{\bb}{\mbox{\bf b}}

\newcommand{\bD}{\mbox{\bf D}}
\newcommand{\bd}{\mbox{\bf d}}
\newcommand{\bH}{\mbox{\bf H}}
\newcommand{\bT}{\mbox{\bf T}}

\newcommand{\bv}{\mbox{\bf v}}
\newcommand{\bV}{\mbox{\bf V}}

\newcommand{\bpi}{\mbox{\boldmath $\pi$}}
\newcommand{\balpha}{\mbox{\boldmath $\alpha$}}
\newcommand{\bdelta}{\mbox{\boldmath $\delta$}}
\newcommand{\bDelta}{\mbox{\boldmath $\Delta$}}
\newcommand{\bbeta}{\mbox{\boldmath $\beta$}}
\newcommand{\btheta}{\mbox{\boldmath $\theta$}}
\newcommand{\boldeta}{\mbox{\boldmath $\eta$}}
\newcommand{\bTheta}{\mbox{\boldmath $\Theta$}}

\def\bu{{\boldsymbol {\u}}}

\newcommand{\be}{\mbox{\bf e}}

\newcommand{\bz}{\mbox{\bf z}}
\newcommand{\bp}{\mbox{\bf p}}
\newcommand{\bZ}{\mbox{\bf Z}}

\def\0{{\bf 0}}
\def\1{{\bf 1}}
\newcommand{\bI}{\mbox{\bf I}}
\newcommand{\bJ}{\mbox{\bf J}}

\newcommand{\simp}{\mathbb{S}}
\newcommand{\real}{\mathbb{R}}
\newcommand{\E}{\mathbb{E}}
\newcommand{\Var}{\mathrm{Var}}

\def\IN{\mathbb{N}}
\def\trans{^{\rm T}}

\newtheorem{theorem}{Theorem}[section]
\newtheorem{corollary}[theorem]{Corollary}
\newtheorem{assumption}[theorem]{Assumption}

\newcommand{\rev}[2][]{#2}
\newcommand{\revnew}[2][]{#2}

\newcommand{\blind}{0}
\begin{document}

\title{\bf Dissecting Microbial Community Structure and Heterogeneity via Multivariate Covariate-Adjusted Clustering}

\if1\blind
{
  \author{}\date{}
  \maketitle
} \fi

\if0\blind
{
  \author{
    Zhongmao Liu$^{1\dagger}$,
    Xiaohui Yin$^{1\dagger}$,
    Yanjiao Zhou$^2$,
    Gen Li$^3$,
    Kun Chen$^{1,2}$\thanks{Corresponding author: kun.chen@uconn.edu}\\
    $^1$\textit{Department of Statistics, University of Connecticut}\\
    $^2$\textit{University of Connecticut Health Center}\\
    $^3$\textit{Department of Biostatistics, University of Michigan}
  }
  \maketitle
}
\fi

\begin{abstract}

\singlespacing
In microbiome studies, it is often of great interest to identify clusters or partitions of microbiome profiles within a study population and to characterize the distinctive attributes of each resulting microbial community. While raw counts or relative compositions are commonly used for such analysis, variations between clusters may be driven or distorted by subject-level covariates, reflecting underlying biological and clinical heterogeneity across individuals. Simultaneously detecting latent communities and identifying covariates that differentiate them can enhance our understanding of the microbiome and its association with health outcomes. To this end, we propose a Dirichlet-multinomial mixture regression (DMMR) model that enables joint clustering of microbiome profiles while accounting for covariates with either homogeneous or heterogeneous effects across clusters. A novel symmetric link function is introduced to facilitate covariate modeling through the compositional parameters. We develop efficient algorithms with convergence guarantees for parameter estimation and establish theoretical properties of the proposed estimators. Extensive simulation studies demonstrate the effectiveness of the method in clustering, feature selection, and heterogeneity detection. We illustrate the utility of DMMR through a comprehensive application to upper-airway microbiota data from a pediatric asthma study, uncovering distinct microbial subtypes and their associations with clinical characteristics.

\noindent Keywords: Count data; Clustering; Heterogeneity; Mixture models.
\end{abstract}

\doublespacing

\section{Introduction}
\label{s:intro}

In microbiome studies, it is often of great interest to identify latent clusters or partitions of the study population based on the observed microbiome data and to characterize the unique microbial profile of each resulting cluster/community \citep{li2015microbiome, zhou2019longitudinal}.
For example, \citet{nakatsu2015gut} cataloged the microbial communities collected from human gut mucosae samples at different stages of colorectal tumorigenesis and concluded that gut metacommunities are associated with the development of colorectal cancer.
\citet{wu2011linking} showed that gut microbiota enterotypes exhibited a strong association with long-term diets, and the correlation mainly existed in protein, animal fat, and carbohydrates.
\citet{zhong2019impact} identified three enterotypes characterized by dominance of different genera in the gut microbiota and further revealed that the correlations between pre-school dietary lifestyle and metabolic phenotypes exhibited a dependency on enterotypes. These studies clearly demonstrate the importance of microbiome clustering, which can facilitate the development of accurate disease diagnostics, targeted interventions, and personalized medicines for improving human health.

Existing methods for microbiome cluster analysis can be categorized into two groups: distance-based methods and model-based methods. Distance-based methods assign samples to different clusters based on pairwise distances or dissimilarity measures between samples, and
many metrics have been proposed to accommodate special features of microbiome data. However, the ``best'' performer is usually context-dependent and the results could be unstable.
With new samples, one usually has to rerun the analysis, but the results could change dramatically. Also, the cluster results typically do not directly provide any insights into ``why'' certain samples form a cluster.

Model-based clustering methods employ a probabilistic model for observed microbiome data; in particular, finite mixture models are quite popular. For instance, \citet{holmes2012dirichlet} proposed Dirichlet-multinomial mixtures (DMM), which model read counts by a multinomial distribution and impose a mixture of Dirichlet distributions as prior for the multinomial parameters.
\citet{mao2022dirichlet} further generalized DMM by incorporating the phylogenetic tree information.
\citet{fang2023clustering} introduced a logistic-normal multinomial mixture model, which substitutes the Dirichlet prior in DMM by a mixture Gaussian prior for the additive log-ratio transformed multinomial parameters.
These model-based approaches can explicitly characterize the ``average'' profile for each cluster and quantify the uncertainty of each sample belonging to each cluster.
New observations can be readily assigned to fitted clusters on the basis of the estimated probabilities.

However, a key limitation of existing clustering methods is that they rarely take into account covariates. In microbiome studies, auxiliary covariates such as demographic and clinical variables are often available and can be associated with the microbial profile. These covariates can either confound clustering results or serve as underlying drivers of microbial heterogeneity. For instance, variables like sex and age are known to influence the composition of human microbiome. Conducting clustering without adjusting for such variables may lead to artificial clusters that reflect these covariates rather than capturing intrinsic, biologically or clinically relevant enterotypes. Moreover, some covariates may exert heterogeneous effects across clusters, helping shape their formation. For example, dietary influences on the microbiome may differ by enterotype, contributing to distinct microbial community structures.

These limitations are particularly relevant in our motivating study on childhood asthma \citep{jackson2018quintupling, zhou2019upper}, where we aim to identify airway microbiome subtypes associated with future risk of exacerbations. In this setting, demographic and clinical factors, such as age, sex, medication use, and baseline symptom severity, may influence microbial patterns or interact with them in predicting disease progression. These considerations underscore the need for clustering methods that not only adjust for covariate effects but also capture potential heterogeneity in their influence across latent microbial communities.

Therefore, in practice, it is important to adjust for common covariate effects as well as to identify heterogeneous effects that contribute to the formation of the clusters, a task referred to as \textit{heterogeneity pursuit} \citep{ZhaoShi2015,li2022pursuing}.

In datasets without clustering structure, a variety of regression methods have been developed to associate microbiome abundances with environmental or biological covariates. Many methods treat each individual taxon as a univariate response and exploit a beta regression \citep{pan2021statistical} or negative binomial regression \citep{zhang2016zero}, which ignore the multivariate and compositional natures of microbiome data.
Multivariate models such as Dirichlet-multinomial regression \citep{chen2013variable,wang2017dirichlet,tang2019zero} and logistic-normal-multinomial regression \citep{xia2013logistic,li2018conditional} are proposed to jointly associate the microbial profile with covariates. However, none of these methods has been adapted for cluster analysis or heterogeneity pursuit.

In this work, we develop a Dirichlet-multinomial mixture regression (DMMR) model with a novel symmetric link function for simultaneously dissecting microbial community structure and heterogeneity induced by individual-level covariates with microbiome count data. Specifically, we model observed microbial read counts using a multinomial distribution, where the probability parameters follow a mixture of Dirichlet distributions. The Dirichlet concentration parameters are further linked to covariates through a regression structure. A key innovation of our approach is the centered log-ratio (clr) link function, which decouples the dispersion parameter from the compositional concentration parameters. This link function operates directly on relative abundances, providing biologically meaningful interpretations and making it well-suited for microbiome data analysis. We further utilize constrained regularization to achieve feature selection (i.e., identifying covariates that are associated with the microbiome) and heterogeneity pursuit (i.e., identifying covariates that have distinct effects in different clusters). Our proposed method achieves estimation consistency as well as variable selection consistency under certain conditions. Overall, the DMMR framework offers a unified statistical model for microbiome data, accounts for uncertainty in clustering, and captures distinct covariate effects on microbial profiles across clusters.

The rest of the paper is organized as follows. In Section~\ref{s:model}, we first introduce the Dirichlet regression and the proposed clr link, and then elaborate the proposed DMMR framework. In Section~\ref{s:computation}, we devise an efficient algorithm for constrained and regularized estimation. We examine the theoretical properties of DMMR in Section~\ref{s:theory}. Comprehensive simulation studies and a real data analysis of the upper-airway microbiota and asthma study are presented in Sections~\ref{s:simulation} and~\ref{s:application}, respectively. Conclusion and future directions are discussed in Section~\ref{s:discuss}.

\section{\texorpdfstring{Dirichlet-Multinomial Mixture Regression\\Framework}{Dirichlet-Multinomial Mixture Regression Framework}}
\label{s:model}

In this section, we first give an overview of the DMM model. We then introduce a new link function for regression modeling with microbiome data based on the centered log-ratio transformation, which serves as a building block for the proposed DMMR framework. Further, we elaborate on the mixture model setup for clustering and introduce regularization with reparameterization for enabling heterogeneity pursuit.

\subsection{Overview of Dirichlet-Multinomial Mixture Model}

Let $\bsm=(m_1,\ldots,m_p)\trans \in \IN^{p}$ be a vector of taxon counts and $M=\sum_{j=1}^{p} m_j$ be the total count. Let $S \in \{1, \dots, K\}$ be the unobservable hidden state variable indicating the cluster membership. Assume $\bsm$, in any given state $S$, follow a Dirichlet-multinomial (DM) distribution. The DM distribution can arise from a compound generation mechanism naturally suited for modeling taxon counts from sequencing reads: a probability vector is generated from a Dirichlet distribution, and a vector of counts is then drawn from a multinomial distribution with the probability vector and the total count. Specifically, it is assumed that
$\bsm \mid (S=k)\sim \mbox{DM}(M,\theta_k,\balpha^{[k]})$, for $k = 1, \ldots, K$,
or, equivalently, the hierarchical structure:
\begin{equation}\label{eq:DMhierarchical}
  \begin{aligned}
    \bsm \mid \bp \sim & \mbox{Multinomial}(M,\bp);\qquad
    \bp \mid (S=k) \sim & \mbox{Dir}(\theta_k,\balpha^{[k]}),
  \end{aligned}
\end{equation}
where $\balpha^{[k]} = (\alpha_{1}^{[k]}, \ldots, \alpha_{p}^{[k]})\trans \in \simp^{p-1}$ is a probability vector\rev{, with $\simp^{p-1}=\{\bp\in\real^p: \sum_{j=1}^p p_j=1, p_j>0, j=1,\ldots,p\}$ being the $p$-dimensional simplex}, $\theta_{k}>0$ is an over-dispersion parameter, and $\bp$ follows a Dirichlet distribution under each state $k$ in the hierarchical formulation. The conditional mean of taxon $j$'s count is $\E(m_{j} \mid S = k) = M\alpha_{j}^{[k]}$. In other words, $\alpha_{j}^{[k]}$ represents the expected relative abundance of the $j$th taxon in the $k$th state. The conditional variance is $\Var(m_j \mid S=k) = M \alpha_{j}^{[k]} (1-\alpha_{j}^{[k]})(M \theta_k + 1)/(\theta_k + 1)$.

The DM model can be extended to the DMM model. Let $\bpi = (\pi_1, \dots, \pi_K)\trans$ be the mixing probability vector of the $K$ clusters and $\btheta = (\theta_1, \dots, \theta_K)\trans$ be a collection of the over-dispersion parameters. Under DMM, it is assumed that $\bsm$ follows a mixture DM distribution with density $\sum_{k=1}^{K}\pi_{k} f_k(\bsm \mid M, \theta_{k}, \balpha^{[k]})$,
where $f_k$ is the density function for $\mbox{DM}(M, \theta_{k}, \balpha^{[k]})$.

\subsection{The ``clr'' Link with Dirichlet Regression}\label{sec:2:dirichlet_reg}

Suppose besides the count or compositional outcomes, a vector of covariates, $\bx\in\real^q$, e.g., demographics, diagnoses, other multi-omics, etc., is also collected. Following the generalized linear model setup, we now consider linking the microbial outcomes to the covariates.

Consider a compositional data vector $\bp\in\simp^{p-1}$.
Assume that $\bp$ follows a Dirichlet distribution $\mbox{Dir}(\theta,\balpha)$, where $\theta>0$ is the over-dispersion parameter and $\balpha\in\simp^{p-1}$ contains the concentration parameters. Some common link functions include the logit link $\log(\alpha_j/\alpha_r)=\beta_{0j}+\bx\trans\bbeta_j$ \citep{yee2010vgam} where $r \in \{1,2,\dots,p\}$ indicates the reference level, and the log-linear link $\log(\alpha_j/\theta)=\beta_{0j}+\bx\trans\bbeta_j$ \citep{chen2013variable, neish2015cluster, wadsworth2017integrative}. Although convenient, these links have significant limitations. The logit link requires a reference, which is usually selected arbitrarily and the resultant parameter estimation lacks symmetry; the log-linear link entangles the over-dispersion parameter and the mean parameters, preventing a direct assessment of the covariate effects on the expected compositions.

We propose an intuitive multivariate link function based on the centered log-ratio transformation that directly links the concentration parameter ${\balpha}$ and the linear predictors.
Specifically, for a compositional vector ${\balpha}\in\simp^{p-1}$, the clr transformation is defined as
\begin{equation*}
\mbox{clr}({\balpha})=\Big(\log \frac{{\alpha}_1}{g({\balpha})}, \ldots, \log \frac{{\alpha}_p}{g({\balpha})}  \Big)\trans,
\end{equation*}
where $g({\balpha})=\Big(\prod_{j=1}^{p}{\alpha}_j\Big)^{1/p}$ is the geometric mean of ${\balpha}$. The transformed vector is in the $p$-dimensional hyperplane subject to $\1_p\trans \mbox{clr}({\balpha}) = 0$, \rev{ where $\1_p$ denotes a vector of ones with length $p$}, which maintains the symmetry and facilitates subsequent modeling.
Based on the clr link, we may consider the following linearly constrained Dirichlet regression formulation,
\begin{equation*}
\bp\sim \ \mbox{Dir}(\theta,{\balpha}),\quad \mbox{clr}({\balpha}) = \bbeta_{0} + \bB\trans\bx, \qquad \mbox{ s.t. } \bbeta_{0}\trans\1_p=0, \bB\1_p=\0,
\end{equation*}
where $\bbeta_0\in\real^p$ is an intercept vector and $\bB=(\bbeta_1,\ldots,\bbeta_p)$ is a $q\times p$ coefficient matrix.

The merit of the proposed clr link function can also be seen from its explicit inverse function, i.e., it implies that the $\alpha_j$s are parameterized by the softmax function,
$$
{\alpha}_j=\frac{\exp{(\beta_{0j}+{\bx}\trans\bbeta_j)}} {\sum_{j'=1}^{p}\exp{(\beta_{0j'}+{\bx}\trans\bbeta_{j'})}}, \ j=1,\ldots,p.
$$
As such, this inverse function is symmetric for all the compositional components, and explicitly characterizes the covariate effects on the expected compositions.

\subsection{Dirichlet-Multinomial Mixture Regression with the ``clr'' Link}\label{sec:2:dmmr}

We propose linearly constrained DMMR model with the clr link. Suppose there are $K$ clusters with mixing probability $\bpi = (\pi_1, \dots, \pi_K)\trans$. Let $\bsm$ be the observed count vector. The DMMR model can then be expressed as
\begin{equation}\label{eq:DMMR}
\begin{aligned}
\Pr(S=k) &= \pi_k; \qquad \bsm \mid (S=k) \sim  \mbox{DM}(M,\theta_k,\balpha^{[k]});\\
\mbox{clr}(\balpha^{[k]}) &= \bbeta_{0}^{[k]}+{\bB^{[k]}}\trans\bx,\quad
 \text{  s.t. }{\bbeta_0^{[k]}}\trans\1_p=0\mbox{ and }\bB^{[k]}\1_p=\0.
\end{aligned}
\end{equation}

The DMMR model is a general framework that integrates the mixture model with the Dirichlet-multinomial regression. The mixture component allows for model-based cluster analysis, and the regression component adjusts for covariate effects.
There are three sets of model parameters: $\{\pi_k\}_{k=1}^K$, $\{\bbeta_{0}^{[k]},\bB^{[k]}\}_{k=1}^K$, and $\{\theta_k\}_{k=1}^K$, where $\{\pi_k\}$ indicates the mixture proportions, $\{\bbeta_{0}^{[k]},\bB^{[k]} = (\bbeta_1^{[k]},\ldots,\bbeta_p^{[k]})\}$ captures the intercepts and covariate effects, and $\{\theta_k\}$ characterizes potential over-dispersion in cluster $k$, for $k=1,\ldots, K$.

Let $\bTheta=\bigl(\{\pi_k\}_{k=1}^K, \{\bbeta_{0}^{[k]}, \bB^{[k]}\}_{k=1}^K, \{\theta_k\}_{k=1}^K\bigr)$ be the collection of all the model parameters. We then have that
\begin{equation}
  f(\bsm; M, \bTheta)= \sum_{k=1}^{K}\pi_k   f_{\mathrm{DM}}(\bsm; M, \bx,\theta_k, \bbeta_{0}^{[k]}, \bB^{[k]}), \label{eq:dmmrdensity}
\end{equation}
where
$$
\alpha_{j}^{[k]}=\frac{\exp{(\beta_{0j}^{[k]} + \bx\trans\bbeta^{[k]}_j)}}{\sum_{j'=1}^{p}\exp{(\beta_{0j'}^{[k]} + \bx\trans\bbeta^{[k]}_{j'})}}, \qquad j = 1, \ldots, p; k = 1,\dots, K.
$$
The Dirichlet-multinomial density $f_{\mathrm{DM}}$ is derived from the hierarchical formulation in \eqref{eq:DMhierarchical}; see Supplement~\ref{supp:1:DM}. \rev{Here the count parameter $M$ is allowed to differ for each observation $\bsm$ and considered known; for simplicity, we may omit $M$ and write the density as $f(\bsm; \bTheta)$.}

The framework subsumes several methods as its special cases. When setting $\bB^{[k]}\equiv\0$ for all $k\in\{1,\ldots, K\}$, the DMMR model reduces to the DMM model without covariate adjustment \citep{holmes2012dirichlet}. Furthermore, when the mixture structure is ignored, the model reduces to a Dirichlet-multinomial regression model incorporating the new clr link function.

\subsection{Feature Selection and Heterogeneity Pursuit}\label{sec:2:heterogeneity}
In microbiome analysis, an important task is to select covariates that exhibit strong associations with the observed data, known as \textit{feature selection}. More intriguingly, in cluster analysis, another related task is to identify covariates that drive heterogeneity among clusters, i.e., \textit{heterogeneity pursuit} \citep{li2022pursuing}.

Let $\bbeta^{[k]}_{(l)}\in\real^{p}$ be the $l$th row of $\bB^{[k]}$, corresponding to the $l$th covariate $x_l$ in the $k$th cluster. By design, the coefficient vector is subject to the linear constraints $\1\trans\bbeta_{(l)}^{[k]}=0, l = 1, \dots, q$.
It is of great interest to distinguish the following three types of covariate effects:
\begin{itemize}
\item[(a)] {\em No effect:} $\bbeta^{[k]}_{(l)}=\0$, for any $k\in\{1,\ldots,K\}$;
\item[(b)] {\em Homogeneous covariate effect:} $\bbeta_{(l)}^{[k]}=\bbeta_{(l)}^{[k^\prime]}\neq\0$, for any pairs of $k,k^\prime\in\{1,\ldots,K\}, \ k\neq k^\prime$;
\item[(c)] {\em Heterogeneous covariate effect:} $\bbeta_{(l)}^{[k]}\neq\bbeta_{(l)}^{[k^\prime]}$, for some $k\neq k^\prime,\ k,k^\prime\in\{1,\ldots,K\}$.
\end{itemize}
Here, (a) indicates that the $l$th covariate does not affect the compositional profile in any cluster and thus is irrelevant; (b) implies that the $l$th covariate plays a role in determining the compositional profile, but its effect is the same across the clusters; (c) shows that the $l$th covariate not only affects the compositional profile, but also its differential effects drive the heterogeneity among different clusters.

\begin{figure}
    \centering
    \includegraphics[width=\linewidth]{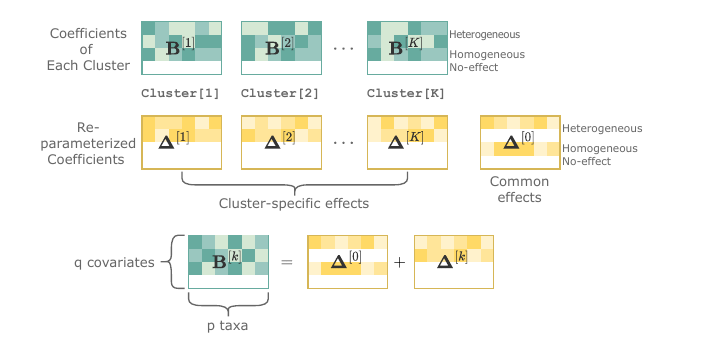}
    \caption{Diagram of coefficient reparameterization in DMMR for heterogeneity pursuit.}
    \label{fig:DMMR_anova}
\end{figure}

Motivated by \citet{li2022pursuing}, we design a regularization scheme for $\{\bB^{[k]}\}_{k=1}^K$ that permits feature selection and heterogeneity pursuit simultaneously.
To achieve this, we first introduce a reparameterization of the original regression coefficients as shown in Figure~\ref{fig:DMMR_anova}. For $k = 1, \ldots, K$ and $l = 1, \ldots, q$, we rewrite  ${\bbeta}^{[k]}_{(l)} \in \real^{p}$ as
\begin{equation}\label{eq:reparam}
\begin{aligned}
{\bbeta}^{[k]}_{(l)}
&= {\bdelta}^{[0]}_{(l)} + {\bdelta}^{[k]}_{(l)},
\quad \text{s.t. }\1\trans{\bdelta}^{[0]}_{(l)}=\1\trans{\bdelta}^{[k]}_{(l)}=0 \text{ and }\sum_{k=1}^{K}{\bdelta}^{[k]}_{(l)}=\0,
\end{aligned}
\end{equation}
where ${\bdelta}^{[0]}_{(l)}$ is the common effect of $x_l$ on the microbiome compositional profile, and ${\bdelta}^{[k]}_{(l)}$ is the cluster-specific effect of $x_l$ in the $k$th cluster.
In particular, $\1\trans{\bdelta}^{[0]}_{(l)}=\1\trans{\bdelta}^{[k]}_{(l)}=0$ ensures that the linear constraints on $\bbeta^{[k]}_{(l)}$ hold, while $\sum_{k=1}^{K}{\bdelta}^{[k]}_{(l)}=\0$ makes the reparameterization identifiable.
Consequently, if ${\bdelta}^{[k]}_{(l)} \neq \0$ for at least one $k \in \{0, 1, \ldots, K\}$, the covariate $x_l$ is deemed effective; if ${\bdelta}^{[k]}_{(l)} \neq \0$ for some $k \in \{1, \ldots, K\}$, the covariate effect is heterogeneous.

As such, DMMR can achieve feature selection and heterogeneity pursuit via a linearly-constrained sparse estimation of $\bdelta^{[k]}_{(l)}$ for $k=0,1,\ldots,K; l=1,\ldots,q$. Denote $\bDelta^{[k]}$ as a $q\times p$ matrix with the $l$th row being $\bdelta^{[k]}_{(l)}$. Now suppose we observe $n$ independent samples of counts, $\bsm_i$, $i  = 1,\ldots, n$. The optimization problem can be expressed as
\begin{equation}\label{eq:optimization}
\begin{aligned}
& \min_{\bTheta} \quad - \frac{1}{n} \sum_{i=1}^{n}\log f(\bsm_i;\bTheta) + \lambda_1 \mathcal{P}_{\gamma}(\bDelta^{[0]}) + \lambda_2 \sum_{k=1}^K \mathcal{P}_{\gamma}(\bDelta^{[k]}); \\
& \quad \qquad \text{s.t.} \quad \bDelta^{[0]}\1=\0,\ \bDelta^{[k]}\1 = \0, \ {\bbeta_{0}^{[k]}}\trans \1 = 0, \ k = 1, \ldots, K;
\sum_{k=1}^K \bDelta^{[k]} = \0,
\end{aligned}
\end{equation}
\rev{where the first term corresponds to the negative log-likelihood function with $f$ defined in \eqref{eq:dmmrdensity}, and the second and the third terms are the regularization terms for the common effects and the heterogeneous effects, respectively. Here $\mathcal{P}_{\gamma}(\cdot)$ is an adaptive sparsity-inducing penalty function with $\gamma$ controlling its adaptivity, and $\lambda_1$ and $\lambda_2$ are tuning parameters. The constraints are inherited from the symmetric link function and the reparameterization; more elaborations are provided in Supplement~\ref{supp:1:constraints}.}

\rev{
Our method aims to cluster the $n$ samples into $K$ latent subgroups, where $K$ is unknown in advance, and to identify covariates with either homogeneous or heterogeneous effects on the microbiome compositional profiles across the clusters. There are many choices of penalties. Here we employ the adaptive group lasso penalty:
\begin{equation} \label{penalty}
 \mathcal{P}_{\gamma}(\bDelta^{[k]}) =  \sum_{l=1}^q w^{[k]}_{(l)} \|\bdelta^{[k]}_{(l)} \|, \qquad k = 0, 1, \dots, K,
\end{equation}
where $\|\cdot\|$ is the $\ell_2$-norm and the adaptive weights are defined as
$w^{[k]}_{(l)}=\|{\bdelta_{(l)}^{[k]}}^{(0)}\|^{-\gamma}$, $\gamma \ge 0$, with ${\bdelta_{(l)}^{[k]}}^{(0)}$ denoting an initial estimator \citep{yuan2006model}. The penalty encourages each parameter vector corresponding to a given covariate and cluster to shrink to zero. As such the covariate is either not effective, homogeneously effective, or heterogeneously effective across all compositions.
\rev{In practice, we set $\gamma = 1$ and compute the adaptive weights as
$w^{[k]}_{(l)} = 1\big/(\|\widehat{\bdelta}_{(l)}^{[k](0)}\| + \epsilon)$,
where $\widehat{\bdelta}_{(l)}^{[k](0)}$ is the DMMR estimate obtained at a near-zero regularization level ($\lambda = 0.001$), and $\epsilon = 10^{-6}$ is a small stabilization constant that prevents the weight from diverging. This choice follows the standard adaptive lasso framework \citep{zou2006adaptive}: groups with large initial estimates receive small penalties and vice versa.} For presentation, we reserve ``DMMR'' for the overall framework and refer to the unpenalized, penalized, and adaptively penalized estimators as DMMR-NP, DMMR-P, and DMMR-AP, respectively.
}

\section{Computation \& Theoretical Guarantee}
\label{s:computation}
\label{s:theory}

We have developed an EM algorithm coupled with the Alternating Direction Method of Multipliers (ADMM) to solve the problem in~\eqref{eq:optimization}. Bayesian Information Criterion (BIC) was used for model selection in numerical studies. See details in Supplement~\ref{supp:computation}.

Recall that $\bTheta$ is defined in Section~\ref{sec:2:dmmr} as the complete set of parameters prior to reparameterization.
Accordingly, we use $\bTheta^\star$ to denote the true parameters of the assumed model.
To facilitates theoretical invesstigation, we formulate the proposed approach as a generalized group lasso problem with special linear constraints.
Specifically, the objective~\eqref{eq:optimization} in Section~\ref{sec:2:heterogeneity}, without the adaptive weights, admits the following general formulation,
\begin{equation} \label{eq:obj_orig_para_group_lasso}
\max_{\bTheta} \quad
\Big\{l(\bTheta) -
n \lambda\sum_{g=1}^{G}
\|\bT_g \bTheta\|_2
\Big\}, \quad s.t. \quad \bD\bTheta = \bd,
\end{equation}
where $l(\cdot)$ is the log-likelihood function such that $l(\bTheta) = \sum_{i=1}^n \log f(\bsm_i\mid \bTheta)$, $\bT_g\bTheta = \bdelta_{(l)}^{[k]}$ for $g = qk+l, G = (K+1)q$ are associated with parameters being penalized for covariates identification, and the expression $\bD\bTheta = \bd$ collects all linear constraints.
For simplicity, we have set $\lambda_1=\lambda_2 = \lambda$. The expressions of $\bT_g$, $\bD$, and $\bd$ are given in Supplement~\ref{supp:proof}.

Theorem \ref{thm:group_lasso_ini} and  Corollary \ref{thm:group_lasso_adp} demonstrate the consistency of our estimators and the zero-sign consistency \citep{she2008sparse} of the adaptive version. Derivations are given in Supplement~\ref{supp:proof}.

\begin{theorem}[Estimation Consistency]
  \label{thm:group_lasso_ini}
  Consider the model with fixed $p,q,K$. Suppose the regularity  conditions in Assumption~\ref{supp:assumption} hold.
  Choose $\lambda=O(n^{-1/2})$.
  Then there exists a local optimizer
  $\widehat{\bTheta}^{(0)}$ of the problem in \eqref{eq:obj_orig_para_group_lasso} such that $\sqrt n(\widehat{\bTheta}^{(0)} - \bTheta^\star) =O_p(1)$.

\end{theorem}
\begin{corollary}[Selection Consistency]
  \label{thm:group_lasso_adp}
  Suppose we use the adaptive group lasso:
  \begin{equation} \label{eq:obj_orig_para_adp_group_lasso}
    \max_{\bTheta} \quad \Big\{l(\bTheta) - n
    \lambda\sum_{g=1}^{G}w_g\\|\bT_g \bTheta\|_2 \Big\}, \quad
    s.t. \quad \bD\bTheta = \bd,
  \end{equation}
  and choose weights based on the non-adaptive estimator in Theorem~\ref{thm:group_lasso_ini}:
  \begin{equation*}
    w_g =\|\bT_g \widehat{\bTheta}^{(0)}\|^{-\gamma}
    = \|\widehat{\bdelta}_{(l)}^{[k]} \|^{-\gamma}.
  \end{equation*}
  Choose $\lambda$ satisfying $\lambda n^{(\gamma+1)/2}\to \infty$ and $\lambda n^{1/2}\to 0$.
  Then there exists a local optimizer $\widehat{\bTheta}_\lambda^\gamma$ of the problem in \eqref{eq:obj_orig_para_adp_group_lasso} such that $\sqrt n(\widehat{\bTheta}_\lambda^\gamma-\bTheta^\star)=O_p(1)$, and $\lim_{n \to \infty}\mathrm{Pr}(\bT_g\widehat \bTheta_\lambda^\gamma=\0) \to 1$ for any group $g$ with $\bT_g \bTheta^\star=\0$.
\end{corollary}

\section{Simulation}
\label{s:simulation}

\subsection{Competing Methods}

We conducted comprehensive simulation studies to evaluate the performance of our proposed DMMR framework \rev{across the non-penalized (DMMR-NP), penalized (DMMR-P), and adaptively penalized (DMMR-AP) variants} from multiple perspectives, including clustering, feature selection, heterogeneity pursuit, and parameter estimation.
To evaluate clustering performance, we considered several baseline approaches that do not incorporate covariate information, including K-Means clustering applied to clr-transformed compositional data (K-Means), hierarchical clustering using Bray-Curtis dissimilarity with complete linkage (HC), and the DMM model without covariate adjustment (DMM).

\subsection{Simulation Setup}

We generated $n$ independent synthetic microbial read counts of dimension $p$ and their associated covariates of dimension $q$ according to DMMR. Specifically, we set $K$ clusters, with cluster probability $\bpi = (1/K) \1_K$ and over-dispersion parameter $\btheta = \theta \1_K$ the same across the clusters. For simplicity, the total read count was fixed at $M=10000$ for all observations. The covariate vector $\bx \in \mathbb{R}^{q}$ was generated from $N_q(\0, \bI)$.
We set the first $q_0$ covariates as relevant covariates, and the first $q_{00}$ covariates from those relevant covariates as heterogeneous covariates.
The $\bdelta_{(l)}^{[k]}$s were generated based on the following mechanism:
\noindent (1) Intercept related $\bbeta_{0}^{[k]}$ for $k=1,\ldots, K$: Each element of $\bbeta_{0}^{[1]}, \ldots, \bbeta_{0}^{[K]}$ was independently sampled from $\mathrm{Unif}(-2,2)$.\\
\noindent (2) Heterogeneity related $\bdelta_{(l)}^{[k]}$ for $l=1,\ldots,q_{00};\, k=1,\ldots, K$:
  \rev{Each element of $\bdelta_{(l)}^{[k]}$ was first independently sampled from $\mathrm{Unif}\{ (-1,-0.5)\cup (0.5,1)\}$. The generated coefficients were then centered across taxa and centered across clusters to satisfy the identifiability constraints. They were further rescaled as
  $\bdelta_{(l)}^{[k]} \gets \bdelta_{(l)}^{[k]} / \sqrt{\sum_{k=1}^K \|\bdelta_{(l)}^{[k]}\|^2} \cdot \sqrt{pK}\, s$, so that $\sqrt{\sum_{k=1}^K \|\bdelta_{(l)}^{[k]}\|^2} = \sqrt{pK}\cdot s$, where $s$ controls the overall signal strength.}\\
\noindent (3) Homogeneity related $\bdelta_{(l)}^{[0]}$ for $l=q_{00}+1,\ldots,q_0$:
  \rev{Each element of $\bdelta_{(l)}^{[0]}$ was first independently sampled from $\mathrm{Unif}\{ (-1,-0.5)\cup (0.5,1)\}$ and then similarly centered and rescaled as above, such that
  it satisfies the zero-sum constraint and $\|\bdelta_{(l)}^{[0]}\| = \sqrt{p}\cdot s$.}\\
\noindent (4) All other $\bdelta_{(l)}^{[k]}$s were set to zeros.

To mimic the real application, we mainly focused on the setting with $n=200$, $K=2$, $p=20$, $q=20$, $q_0=10$, and $q_{00}=5$. The dispersion parameter $\theta$ took values in $\{0.05, 0.1\}$, and the signal strength was set as \rev{$s \in \{0.2, 0.4, 0.6 \}$}.

\rev{To further assess scalability and stability, we additionally examined settings with varying read counts, higher-dimensional covariates ($q \in \{50, 100\}$), and a larger number of clusters ($K = 5$); the detailed settings and corresponding results are reported in Supplement \ref{s:supp_sim}.}

\subsection{Evaluation Metrics}

For clustering, we assessed the accuracy of selecting the true number of clusters ($K$) using the corresponding model-selection criterion for each method; see details in Supplement~\ref{supp:selection}. When $K$ was identified, we further evaluated clustering accuracy using Cohen's kappa coefficient, with cluster labels optimally aligned to mitigate the impact of label switching.

For variable selection, we reported sensitivity (true positive rate, TPR), specificity (true negative rate, TNR), and the $F_1$ score for selecting both relevant covariates and heterogeneous covariates. \rev{Sensitivity is defined as the proportion of truly active covariates correctly identified, whereas specificity is the proportion of truly inactive covariates correctly classified as such. Because the trade-off between sensitivity and specificity is often unavoidable, we additionally reported the $F_1$ score, the harmonic mean of precision and sensitivity.}

\rev{For parameter estimation, we quantified accuracy via the \revnew{relative root-mean-squared errors} (rMSE) between the estimated parameters and their true values. \revnew{For a generic parameter $\boldeta$ with entries $\{\eta_j\}$ and estimate $\widehat\boldeta$, the rMSE is defined as $\big(\sum_{j}(\widehat\eta_j-\eta_j)^2 \big/ \sum_{j}\eta_j^2\big)^{1/2}$, where the summation runs over all entries of the parameter.} We reported rMSE for the mixing proportions $\bpi$, the dispersion parameter $\btheta$, and the regression coefficients $\bB$ and $\bDelta$.}

\revnew{The simulation under each setting is repeated 200 times, and we report means
and standard deviations of the metrics across all replications.}

\subsection{Simulation Results}

\subsubsection{Clustering Performance}

Table~\ref{tab:selection_Kappa} summarizes the baseline clustering results for $K=2$ and $q=20$;
\rev{additional clustering results for varying read counts, higher-dimensional covariates, and a larger number of clusters are reported in Supplement Table~\ref{sim:tab:all_K_selection_add}.}

When the covariate signal was weak ($s=0.2$), clustering accuracy remained high for the model-based methods, with DMMR-NP and DMMR-AP generally outperforming DMMR-P. As the signal increased ($s=0.4$ and $s=0.6$), all DMMR variants achieved near-perfect clustering, whereas methods that ignored covariates, such as K-Means, HC, and DMM, degraded substantially. Overall, stronger covariate signals made the clustering structure more apparent for DMMR, highlighting the benefit of covariate adjustment and regularization.

The over-dispersion parameter $\theta$ also affected clustering performance. Stronger over-dispersion increases variability in observed counts, and as expected, makes cluster recovery more challenging.
The impact of over-dispersion was not uniform: the distance-based methods (K-Means and HC) and the unpenalized DMMR-NP tended to be more sensitive, whereas DMMR-P and especially DMMR-AP were generally more stable.

\rev{The results in Supplement Table~\ref{sim:tab:all_K_selection_add} showed that the above observations were largely preserved under more complex settings. When sequencing depths moderately varied across samples, clustering performance remained stable. In higher-dimensional settings, the DMMR variants continued to achieve strong clustering performance and generally compared favorably with methods that ignored covariate information.
When the number of clusters increased to $K=5$, clustering became more challenging overall, and the benefit of regularization became more apparent: DMMR-P and especially DMMR-AP achieved substantially higher Acc($K$) and stronger label agreement than DMMR-NP.}

\begin{table}
  \centering
  \caption{Simulation: Accuracy of selecting $K$ (Acc($K$)) and the Kappa statistics (Kappa)}
  \vspace{-1em}
  \label{tab:selection_Kappa}
  \begin{tabular}{lrrrr}
    \toprule
    & \multicolumn{2}{c}{$\theta = 0.05$} & \multicolumn{2}{c}{$\theta = 0.10$} \\
    \cmidrule(l{3pt}r{3pt}){2-3} \cmidrule(l{3pt}r{3pt}){4-5}
    & Acc($K$) & Kappa & Acc($K$) & Kappa\\
    \midrule
    & \multicolumn{4}{c}{$s = 0.2$} \\
    K-Means & 1.000 & 0.991 (0.014) & 0.985 & 0.986 (0.016)\\
    HC & 0.875 & 0.907 (0.104) & 0.675 & 0.841 (0.103)\\
    DMM & 1.000 & 0.992 (0.014) & 1.000 & 0.985 (0.020)\\
    DMMR-NP & 0.995 & 1.000 (0.001) & 0.735 & 0.999 (0.004)\\
    DMMR-P & 0.600 & 0.995 (0.011) & 0.820 & 0.986 (0.018)\\
    DMMR-AP & 0.945 & 1.000 (0.002) & 0.920 & 0.995 (0.011)\\
    \midrule
    & \multicolumn{4}{c}{$s = 0.4$} \\
    K-Means & 0.735 & 0.945 (0.049) & 0.575 & 0.944 (0.040)\\
    HC & 0.175 & 0.730 (0.222) & 0.075 & 0.756 (0.178)\\
    DMM & 0.950 & 0.884 (0.118) & 0.965 & 0.878 (0.107)\\
    DMMR-NP & 1.000 & 1.000 (0.002) & 0.995 & 0.999 (0.003)\\
    DMMR-P & 0.965 & 0.999 (0.003) & 0.975 & 0.998 (0.005)\\
    DMMR-AP & 1.000 & 1.000 (0.001) & 0.995 & 0.999 (0.003)\\
    \midrule
    & \multicolumn{4}{c}{$s = 0.6$} \\
    K-Means & 0.080 & 0.881 (0.079) & 0.035 & 0.868 (0.078)\\
    HC & 0.015 & 0.154 (0.111) & 0.005 & 0.324 (-)\\
    DMM & 0.805 & 0.452 (0.262) & 0.580 & 0.519 (0.246)\\
    DMMR-NP & 0.960 & 1.000 (0.001) & 0.935 & 1.000 (0.002)\\
    DMMR-P & 0.905 & 1.000 (0.002) & 0.890 & 0.999 (0.004)\\
    DMMR-AP & 0.965 & 1.000 (0.001) & 0.935 & 0.995 (0.068)\\
    \bottomrule
  \end{tabular}
\end{table}

\subsubsection{Feature Selection}

Table~\ref{tab:selection_relevant_main} summarizes the baseline relevant-covariate selection results, and Supplement Table~\ref{sim:tab:selection_bic_heterogeneous} summarizes the heterogeneous-covariate selection results \revnew{(i.e., how well covariates with cluster-specific effects are identified)}, for $K=2$ and $q=20$; \rev{additional results
are reported in Supplement Tables~\ref{sim:tab:selection_bic_relevant_add} and \ref{sim:tab:selection_bic_heterogeneous_add}.}

Across most scenarios, DMMR-AP substantially improved variable-selection accuracy compared with DMMR-P. This improvement was particularly pronounced when the covariate signal was weak ($s=0.2$).
As the signal strength increased ($s=0.4$ and $s=0.6$), both methods identified most true signals, but DMMR-AP still provided a more favorable balance between sensitivity and specificity.
\rev{Overall, these results showed that adaptive weighting strengthened the heterogeneity-pursuit component of DMMR while better preserving clustering accuracy.}

\begin{table}
  \caption{Simulation: Relevant covariates selection performance when $K=2$, $q=20$, and $M_i$ is fixed across samples.}
  \vspace{-1em}
  \label{tab:selection_relevant_main}
  \centering
  \begin{tabular}{llccc}
  \toprule
  $s$ & Model & Sensitivity & Specificity & $F_1$\\
  \midrule
  \multicolumn{5}{c}{$\theta = 0.05$} \\
  \multirow{2}{*}{0.2} & DMMR-P & 0.28 (0.34) & 1.00 (0.01) & 0.34 (0.38)\\
    & DMMR-AP & 1.00 (0.00) & 0.98 (0.05) & 0.99 (0.02)\\
  \multirow{2}{*}{0.4} & DMMR-P & 1.00 (0.00) & 0.75 (0.29) & 0.90 (0.10)\\
    & DMMR-AP & 1.00 (0.00) & 1.00 (0.02) & 1.00 (0.01)\\
  \multirow{2}{*}{0.6} & DMMR-P & 0.98 (0.10) & 0.49 (0.39) & 0.80 (0.12)\\
    & DMMR-AP & 0.99 (0.07) & 0.99 (0.08) & 0.99 (0.05)\\
  \midrule
  \multicolumn{5}{c}{$\theta = 0.10$} \\
  \multirow{2}{*}{0.2} & DMMR-P & 0.04 (0.11) & 1.00 (0.00) & 0.06 (0.15)\\
    & DMMR-AP & 0.79 (0.40) & 0.98 (0.04) & 0.78 (0.39)\\
  \multirow{2}{*}{0.4} & DMMR-P & 1.00 (0.04) & 0.84 (0.22) & 0.93 (0.08)\\
    & DMMR-AP & 1.00 (0.04) & 0.99 (0.04) & 0.99 (0.03)\\
  \multirow{2}{*}{0.6} & DMMR-P & 0.97 (0.12) & 0.57 (0.38) & 0.82 (0.12)\\
    & DMMR-AP & 0.99 (0.08) & 0.97 (0.13) & 0.98 (0.07)\\
  \bottomrule
  \end{tabular}
\end{table}

\subsubsection{Parameter Estimation}

\rev{Supplement Table~\ref{sim:tab:selection_bic_mse} summarizes the baseline estimation results \revnew{(i.e., the accuracy of estimating the model parameters)} for $K=2$ and $q=20$; additional results
are reported in Supplement Table~\ref{sim:tab:all_mse_add}.}

The mixing proportions $\bpi$ were estimated with similar accuracy across methods. For the regression parameters $\bB$ and $\bDelta$, DMMR-AP consistently achieved the smallest rMSE, while DMMR-P usually improved upon DMMR-NP; the gap between the latter two was smallest when the signal was weak ($s=0.2$). The dispersion parameter $\btheta$ was best estimated by DMMR-AP, but overall its estimation was less stable than that of the regression coefficients. This is likely because the estimation is sensitive to how the model allocates variability between fitted covariate effects and residual over-dispersion \citep{zou2006adaptive}.

\rev{The same pattern held under heterogeneous read counts and became more pronounced in harder settings. In Supplement Table~\ref{sim:tab:all_mse_add}, DMMR-NP deteriorated markedly when $q$ or $K$ increased, while DMMR-AP remained the most accurate estimator for $\bB$ and $\bDelta$.}

\section{\texorpdfstring{Application to the Upper-airway Microbiota and\\Asthma Study on Children}{Application to the Upper-airway Microbiota and Asthma Study on Children}}
\label{s:application}

\vspace{-0.5em}

Exacerbations of asthma impose a significant burden on children, their families, and the healthcare system, and can contribute to long-term declines in lung function. A critical phase in asthma management is the ``Yellow Zone'' (YZ), which is a period of early loss of asthma control during which patients are at elevated risk of progression to severe exacerbation.

The Step-Up Yellow Zone Inhaled Corticosteroids to Prevent Exacerbations (STICS) trial \citep{jackson2018quintupling} studied school-aged children with mild-to-moderate persistent asthma and evaluated whether quintupling inhaled corticosteroids at YZ onset could prevent severe exacerbation. Children aged 5 to 11 years received low-dose inhaled corticosteroids for 48 weeks and were then randomized to continue the same dose or to quintuple it during YZ episodes. They found no significant reduction in severe exacerbations under dose escalation.

\citet{zhou2019upper} analyzed a subset of 214 children from the STICS trial, to investigate the role of the upper-airway microbiota in asthma control. For those participants, nasal swab samples were collected at two clinically relevant time points: (1) at the randomization visit (RD), when participants were asymptomatic, and (2) at the onset of the first YZ episode, prior to administration of the assigned YZ intervention. Total genomic DNA was extracted from
nasal blow samples, and the V1-V3 regions of the 16S rRNA gene were sequenced to generate a taxonomic profile. Reads were rarefied to 10,000 per sample and aggregated at the genus level.
Their hierarchical clustering found that airway microbiota colonization patterns are differentially associated with risk of loss of asthma control and severe exacerbation.

Motivated by \citet{zhou2019upper}, we reanalyzed this cohort to identify upper-airway microbiota clusters and assessed whether they were associated with future YZ episodes. In contrast to \citet{zhou2019upper}, DMMR incorporates demographic and clinical variables directly into the clustering model, allowing us to adjust for covariate effects and detect heterogeneity in how these factors relate to latent microbial communities.

\subsection{DMMR Analysis}
\rev{We applied the adaptively penalized DMMR estimator (DMMR for simplicity)} to the STICS data to identify upper-airway microbiome clusters and the demographic or clinical features that differentiate them.
The DMMR model was fitted to the upper-airway microbiota profiles of 214 subjects with 18 covariates, including age, BMI, number of oral steroid courses, IgE, gender, ethnicity, race, parental history of asthma, smoke exposure, pet exposure, eczema indicators, steroid use, antibiotic use, and viral infection status; categorical variables were expanded into dummy variables (Supplement Table~\ref{tab:app:covariates}). For the microbiome counts, we retained taxa with relative abundance of at least $0.5\%$ and grouped the remaining low-abundance taxa into ``Other'', yielding $p = 24$ taxa.
\rev{To assess sensitivity to this abundance threshold, we repeated the analysis with a lower cutoff of $0.1\%$, which retained $p = 41$ taxa.}

For initialization, we applied hierarchical clustering (HC) to the genus-level microbiome profiles using complete linkage and Bray--Curtis dissimilarity. The resulting cluster proportions initialized $\bpi$, and DM models fitted within each HC cluster provided initial estimates of $\btheta$ and $\bB^{[k]}$ for $k=1,\dots,K$. We refer to this initialization procedure as ``HC+DM''.

\subsection{Results}\label{app:results}

The number of clusters was selected as $K = 3$.
We assigned cluster names post hoc according to the dominant relative abundance patterns identified by DMMR: \textit{Strep-dominant}, \textit{Dolo/Coryne-dominant}, and \textit{Mixed-pathobiont}.

Table~\ref{app:tab:pi_theta} in Supplement \ref{supp:app:para} presents the estimated mixing proportions ($\bpi$) and the over-dispersion parameters ($\btheta$) for HC+DM and DMMR. The two approaches produced different cluster profiles. With HC+DM, the \textit{Dolo/Coryne-dominant} cluster was the largest (54\%), followed by the \textit{Strep-dominant} cluster (29\%), whereas DMMR assigned 52\% of subjects to the \textit{Strep-dominant} cluster, 33\% to the \textit{Dolo/Coryne-dominant} cluster, and 15\% to the \textit{Mixed-pathobiont} cluster. HC+DM also yielded larger over-dispersion estimates, suggesting greater within-cluster variability, whereas DMMR produced more coherent clusters.

Based on the fitted model parameters, each observation was assigned to a cluster using Bayes' rule. The relative abundances of the $24$ taxa are visualized in Figure~\ref{app:fig:rela_abund} through heatmaps. In these heatmaps, columns represent individual samples and rows represent microbial taxa, with cell color intensity indicating the normalized abundance of each taxon within a sample. \rev{These visualizations facilitated comparison of microbiome composition across clusters and revealed distinct taxonomic signatures. The central alluvial plot further showed that most discrepancies arose from samples reassigned between \textit{Strep-dominant} and \textit{Dolo/Coryne-dominant} clusters, while the \textit{Mixed-pathobiont} cluster remained largely consistent across the two approaches. To facilitate comparison, cluster labels were aligned across methods. For HC, the cluster labels were permuted to best align with the DMMR partition before applying the same descriptive names. Thus, the labels facilitate comparison but do not imply identical partitions or identical dominant-taxa structures across methods. As expected, the resulting partitions differed between DMMR and HC. Quantitatively, the Adjusted Rand Index (ARI) \citep{hubert1985comparing} between the two clustering results was $0.14$, indicating limited agreement between the two methods. These differences reflected the distinct clustering principles: HC is based on geometric dissimilarity, whereas DMMR uses a probabilistic model that incorporates covariate heterogeneity.}

\begin{figure}
    \centering
         \includegraphics[width=\linewidth]{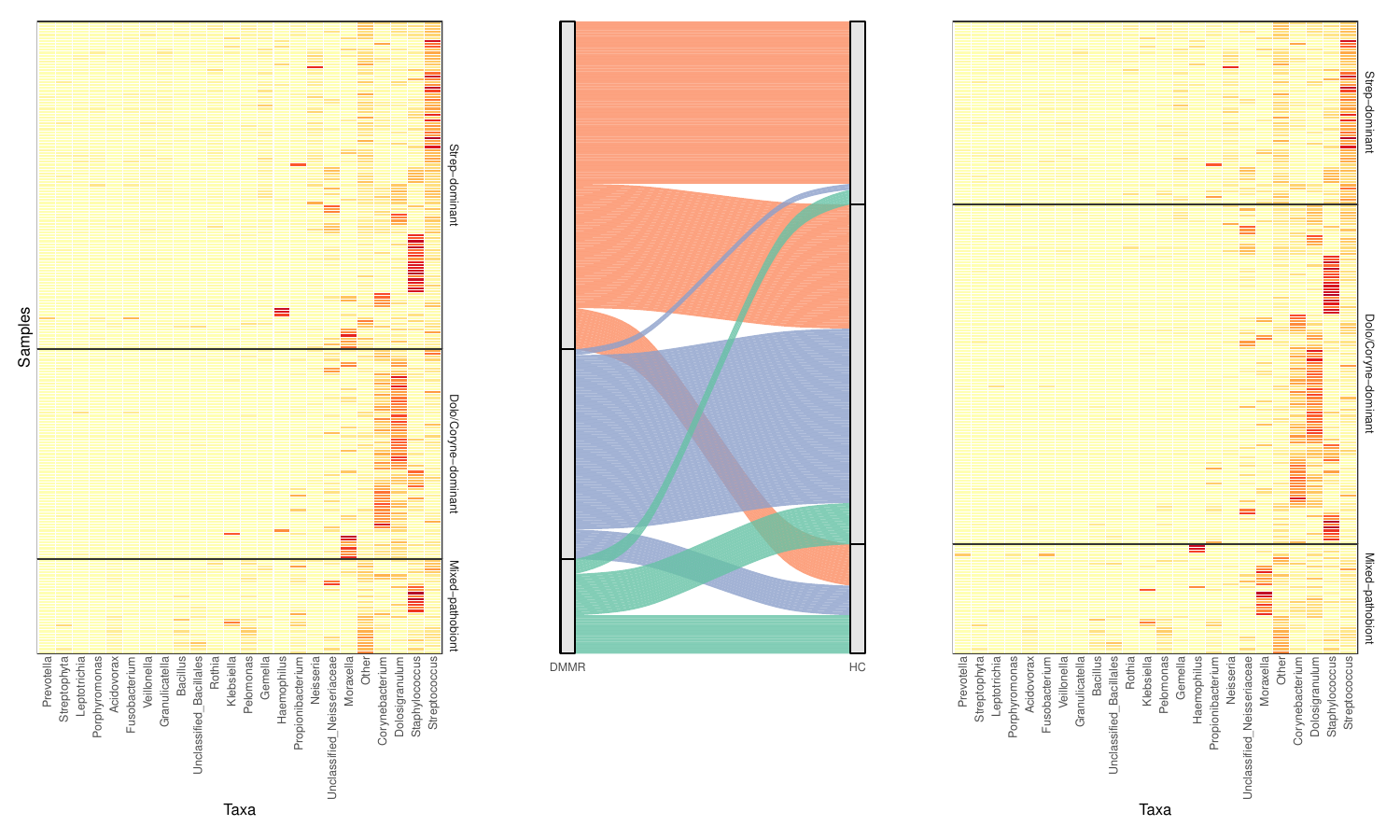}
    \caption{Application: Relative abundance heatmaps of taxa across clusters identified by DMMR (on the left) and HC (on the right). The central alluvial plot illustrates how individual samples correspond between clusters across the two approaches.}
    \label{app:fig:rela_abund}
\end{figure}

To investigate the sources of heterogeneity among the three identified clusters, we examined the estimated covariate coefficients from the DMMR model in Figure~\ref{app:fig:coef}. A total of $18$ relevant covariates were identified, including $15$ with homogeneous effects and $3$ with heterogeneous effects. \revnew{The $15$ homogeneous covariates have effects shared across all three clusters, whereas the $3$ heterogeneous covariates, namely, ``parent\_ast8'', ``eczema1'', and ``steroid1'', have cluster-specific effects that differ across clusters. Note that the cluster-specific effect of ``parent\_ast8'' is shrunk to exactly zero in the \textit{Mixed-pathobiont} cluster by the cluster-specific sparsity penalty.} For example, both ``eczema1'' and ``steroid1'' had significantly different coefficients across the three clusters. Focusing on their effects on the Staphylococcus and Streptococcus taxa, the results indicated that having eczema or using steroids was associated with a higher proportion of Staphylococcus and a lower proportion of Streptococcus within the \textit{Mixed-pathobiont} cluster. Conversely, in the \textit{Dolo/Coryne-dominant} and \textit{Strep-dominant} clusters, these covariates are associated with a lower abundance of Staphylococcus and a higher abundance of Streptococcus.
\rev{These patterns were supported by previous studies. \citet{azev2023nasal} reported that children with atopic dermatitis exhibit significantly increased nasal colonization by Staphylococcus aureus, suggesting that host inflammatory conditions can shape upper airway microbial communities. \citet{hartmann2021effects} demonstrated that corticosteroid treatment can significantly alter respiratory microbiome composition, including shifts in major phyla and specific taxa such as Streptococcus and Neisseria. These findings provided biological support for the association between eczema or using steroids and the cluster-specific microbial patterns observed in the STICS analysis.}

\begin{figure}
    \centering
    \includegraphics[width=\linewidth]{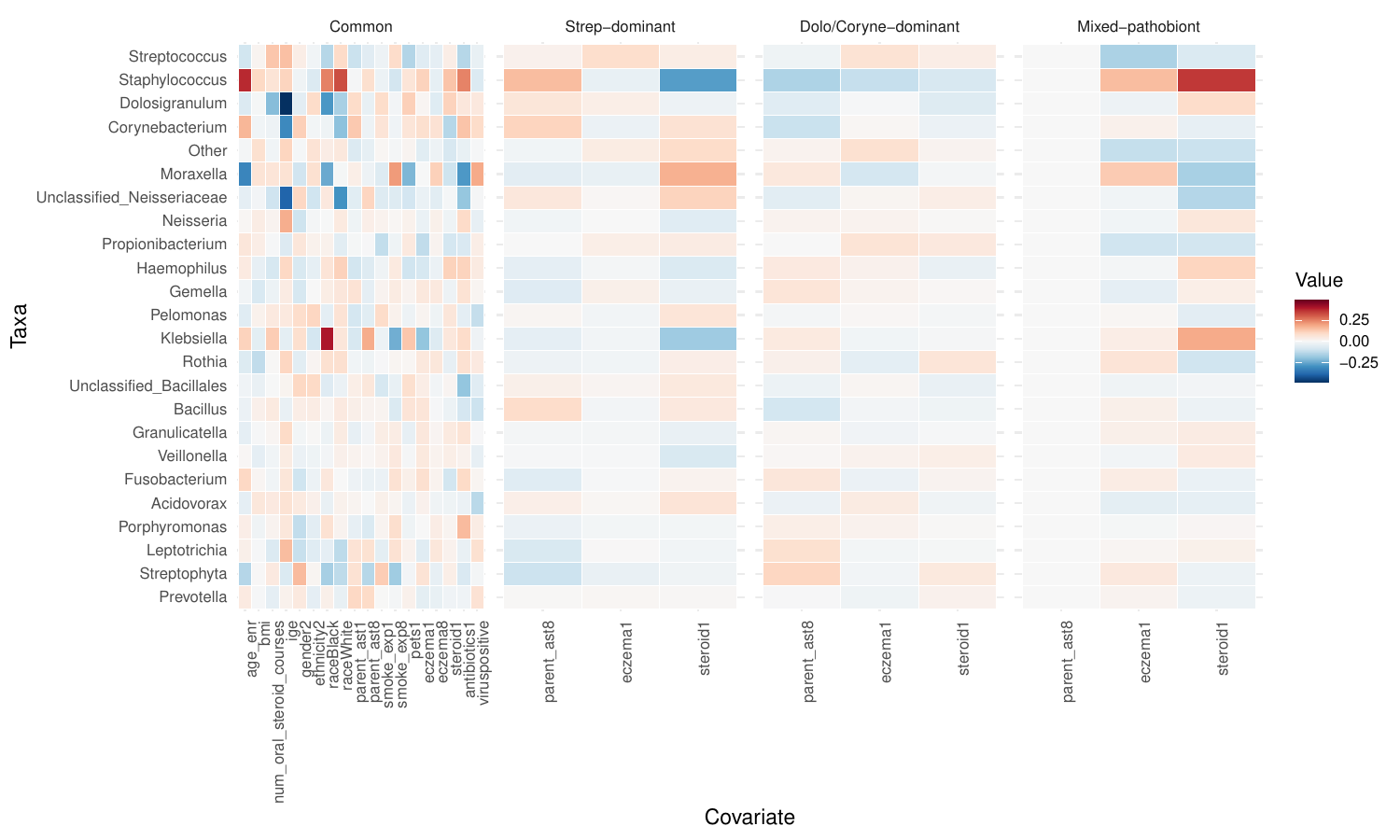}
    \caption{\rev{Application: Estimated covariate effects for the three clusters identified by DMMR. The common effects of all covariates and the cluster-specific effects of the selected heterogeneous covariates are illustrated.} \revnew{Among the three heterogeneous covariates, ``parent\_ast8'' has its cluster-specific effect shrunk to zero in the \textit{Mixed-pathobiont} cluster.}
    }
    \label{app:fig:coef}
\end{figure}

Additional results on model estimation, sensitivity analysis to the taxa abundance threshold, and the association between cluster membership and time to YZ episode are reported in Supplement~\ref{s:supp_application}. Overall, these results indicated that DMMR was reasonably stable with respect to the abundance threshold and led to distinct time-to-event distributions of YZ episodes across clusters.

\section{Discussion}
\label{s:discuss}

We have developed the Dirichlet-Multinomial Mixture Regression (DMMR) approach for clustering microbiome samples while adjusting for covariate effects and identifying sources of heterogeneity. By employing a novel symmetric link function, DMMR is capable of detecting subpopulations, identifying covariates/features that exhibit heterogeneous, common, or null effects on cluster formation, and highlighting key taxa with distinct abundance patterns across clusters. We have developed an R package that implements the proposed methodology.

There are several potential directions for future research. First, it is interesting to extend our theoretical results under a more general large data framework. Second, the concept of heterogeneity pursuit can be applied to other mixture models for clustering purposes. While the DM distribution is utilized here to model microbiome data due to its ability to handle over-dispersion, the idea can be applied to a variety of distributions, including Negative multinomial, multivariate normal, Poisson, or any other distributions that permit the characterization of common effects among clusters and cluster-specific effects. Third, with longitudinal microbiome data,  not only may there exist a set of microbial states, but also transitions between states over time may happen.  These microbial states themselves and the dynamic transitions of microbial states may have associations with disease or health status. Some related studies can be found in \citet{zhou2019upper,xiong2015changes,lee2017bacterial}. Our heterogeneity pursuit methods could be adapted to a unified Hidden Markov Model (HMM) to study microbial states and their transitions over time.

\section*{Acknowledgements}

The authors thank the reviewers and the Associate Editor for their very helpful suggestions.

\section*{Data Availability}

The data that support the findings in this paper are available from \citet{zhou2019upper} at \url{https://www.nature.com/articles/s41467-019-13698-x}.

\bibliography{mybibilo}{}
\appendix
\renewcommand{\thetable}{S\arabic{table}}
\setcounter{table}{0}
\renewcommand{\thefigure}{S\arabic{figure}}
\setcounter{figure}{0}
\renewcommand{\theequation}{S\arabic{equation}}
\setcounter{equation}{0}
\renewcommand{\theassumption}{S\arabic{theorem}}
\setcounter{theorem}{0}
\section{Additional Details on DMMR Formulation}

\subsection{Dirichlet-Multinomial Distribution}\label{supp:1:DM}

In the following we provide the explicit form of the density function of the Dirichlet-multinomial distribution, which is used to model the observed count data in our proposed DMMR framework. The Dirichlet-multinomial distribution can be derived by marginalizing out the latent probability vector $\bp$ from a hierarchical model where $\bp$ follows a Dirichlet distribution and the observed counts $\bsm$ follow a multinomial distribution conditional on $\bp$. Specifically, we have
    \begin{equation*}
      \begin{aligned}
        \bsm \mid \bp &\sim \mbox{Multinomial}(M,\bp), \\
        \bp \mid (S = k) &\sim \ \mbox{Dir}( \balpha^{[k]} / \theta_k).
      \end{aligned}
    \end{equation*}
The density functions of the Dirichlet distribution and multinomial distribution are given by
    \begin{equation*}
      \begin{aligned}
        f_{\mathrm{Dir}} (\bp \mid S = k; \theta_k, \balpha^{[k]}) & = \frac{\Gamma(1/\theta_k)}{ \prod_{j=1}^p \Gamma(\alpha_j^{[k]}/\theta_k) } \prod_{j=1}^p p_j^{\alpha_j^{[k]}/\theta_k - 1}; \\
        f_{\mathrm{Multi}}(\bsm \mid \bp; M, \theta_k, \balpha^{[k]}) & = \frac{ \Gamma(M+1) }{ \prod_{j=1}^p \Gamma(m_j+1) }\prod_{j=1}^p p_j^{m_j}.
      \end{aligned}
    \end{equation*}
    Combining these two components, we obtain the density function of the Dirichlet-multinomial distribution as
    \begin{equation*}
      \begin{aligned}
        f_{\mathrm{DM}} (\bsm \mid S = k; M, \theta_k, \balpha^{[k]})
        &= \frac{ \Gamma(M+1) \Gamma(1/\theta_k) }{ \prod_{j=1}^p \Gamma(m_j+1) \Gamma(\alpha_j^{[k]}/\theta_k)}\int \prod_{j=1}^p p_j^{ m_j + \alpha_j^{[k]}/\theta_k - 1} \,d\bp\\
        &= \frac{ \Gamma(M+1) \Gamma(1/\theta_k) }{ \prod_{j=1}^p \Gamma(m_j+1) \prod_{j=1}^p \Gamma(\alpha_j^{[k]}/\theta_k)} \frac{\prod_{j=1}^p \Gamma(m_j + \alpha_j^{[k]}/\theta_k) }{\Gamma(M + 1/\theta_k)}\\
        &= \frac{\Gamma(M+1)}{\prod_{j=1}^p \Gamma(m_j +1)} \frac{\prod_{j=1}^p \prod_{l=0}^{m_j - 1}(\alpha_{j}^{[k]} + l\theta_k)}{\prod_{l=0}^{M-1}(1+l\theta_k)}.
      \end{aligned}
    \end{equation*}

\subsection{Further Explanation of the Constraints and Penalties}\label{supp:1:constraints}

The objective function of the DMMR model can be expressed as
\begin{equation}
\begin{aligned}
& \min_{\bTheta} \quad - \frac{1}{n} \sum_{i=1}^{n}\log f(\bsm_i;\bTheta) + \lambda_1 \mathcal{P}_{\gamma}(\bDelta^{[0]}) + \lambda_2 \sum_{k=1}^K \mathcal{P}_{\gamma}(\bDelta^{[k]}); \\
& \quad \qquad \text{s.t.} \quad \bDelta^{[0]}\1=\0,\ \bDelta^{[k]}\1 = \0, \ {\bbeta_{0}^{[k]}}\trans \1 = 0, \ k = 1, \ldots, K;
\sum_{k=1}^K \bDelta^{[k]} = \0,
\end{aligned}
\end{equation}
where the first term corresponds to the negative log-likelihood function with $f$ defined in \eqref{eq:dmmrdensity}, and the second and the third terms are the regularization terms for the common effects and the heterogeneous effects, respectively. Here $\mathcal{P}_{\gamma}(\cdot)$ is an adaptive sparsity-inducing penalty function with $\gamma$ controlling its adaptivity, and $\lambda_1$ and $\lambda_2$ are tuning parameters. The constraints are inherited from the symmetric link function and the reparameterization.

We provide further explanations of the constraints and penalties in the optimization problem above.

    \begin{itemize}
    \item
      These constraints arise from two procedures:
      \begin{enumerate}
      \item
        Given the clr-based transformation,
        $$
        \alpha_{j}^{[k]}=\frac{\exp{(\beta_{0j}^{[k]} + \bx\trans\bbeta^{[k]}_j)}}{\sum_{j'=1}^{p}\exp{(\beta_{0j'}^{[k]} + \bx\trans\bbeta^{[k]}_{j'})}}, \qquad j = 1, \ldots, p;\, k = 1,\dots, K.
        $$
        If we set $\widetilde \bbeta_0^{[k]} = \bbeta_0^{[k]} + \delta\1$, i.e. $\widetilde \bbeta_{0j}^{[k]} = \bbeta_{0j}^{[k]} + \delta$ for $j=1,\dots,p$, then $\alpha_j^{[k]}$ remains invariant to such a linear shift. Similarly, if we set $\widetilde \bbeta_{(l)}^{[k]} =\bbeta_{(l)}^{[k]} + \delta\1$, i.e. $\widetilde \bbeta_{lj}^{[k]} = \bbeta_{lj}^{[k]} + \delta$ for $j=1,\dots,p$, the resulting $\alpha_{j}^{[k]}$ is also invariant to this shift because
        $$
        \widetilde\alpha_{j}^{[k]}=\frac{\exp{(\beta_{0j}^{[k]} + \bx\trans\bbeta^{[k]}_j + x_j\delta)}}{\sum_{j'=1}^{p}\exp{(\beta_{0j'}^{[k]} + \bx\trans\bbeta^{[k]}_{j'} + x_j\delta)}} = \alpha_{j}^{[k]}.
        $$
        To ensure identifiability of the parameters $\{\bbeta_0^{[k]}, \bB^{[k]}\}_{k=1}^K$, we therefore impose the constraints ${\bbeta_0^{[k]}}\trans\1=0$ and ${\bbeta_{(l)}^{[k]}}\trans\1=0$ for $l=1,\dots, q$, i.e.
        $$
        {\bbeta_0^{[k]}}\trans\1=0\text{ and }\bB^{[k]}\1=\0.
        $$

      \item
        We reparameterize the coefficient for the $l$th covariate to distinguish the common-effect shared across all clusters from the cluster-specific effect as
        $$
        {\bbeta}^{[k]}_{(l)} = {\bdelta}^{[0]}_{(l)} + {\bdelta}^{[k]}_{(l)},  \quad k=1,\dots,K,
        $$
        where ${\bdelta}^{[0]}_{(l)}$ represents the common effect and ${\bdelta}^{[k]}_{(l)}$ captures the deviation of cluster $k$ from this common baseline.
        To ensure a unique ANOVA-type decomposition, we impose the centering constraint
        $$ \sum_{k=1}^{K}{\bdelta}^{[k]}_{(l)}=\0, $$
        which forces the cluster-specific effects to sum to zero across clusters.
        Moreover, to maintain the previous identifiability condition ${\bbeta_{(l)}^{[k]}}\trans\1=\0$, we require that
        $$\1\trans{\bdelta}^{[0]}_{(l)}=\1\trans{\bdelta}^{[k]}_{(l)}=0.$$

      \end{enumerate}
      These constitute the four constraints.

    \item We allowed different sparsity structures for the common effect and the cluster-specific effects. Biologically, some covariate-taxon associations are conserved across clusters, whereas others appear only within particular clusters. Accordingly, the sparsity pattern (which coefficients are nonzero) and the sparsity level (how many are nonzero) may differ between the common and cluster-specific effects. To accommodate these differences, we generally apply separate penalties and tuning parameters, and, when appropriate, use different norms (e.g., Frobenius-norm penalties for improving estimability, group penalties for identifying active covariates, and elementwise $\ell_1$ penalties for detecting covariate-taxon associations). We list several types of penalty functions below:
      \begin{enumerate}
        \item $ \mathcal{P}_{\gamma}(\bDelta^{[k]}) =  \|\bDelta^{[k]}\|_F$, a ridge-type (Frobenius-norm) penalty used to improve estimability and stabilize estimation.
        \item $ \mathcal{P}_{\gamma}(\bDelta^{[k]}) =  \sum_{l=1}^q  \|\bdelta^{[k]}_{(l)} \|$, a group-type penalty designed to identify covariates that have effects across taxa.
        \item $ \mathcal{P}_{\gamma}(\bDelta^{[k]}) =  \sum_{l=1}^q  \|\bdelta^{[k]}_{(l)} \|_1$, an elementwise $\ell_1$ penalty that induces entry-wise sparsity for covariate-taxon associations.
      \end{enumerate}
    \end{itemize}

\section{Details on Computation}\label{supp:computation}

\subsection{Computational Algorithm}

We develop an EM algorithm coupled with the Alternating Direction Method of Multipliers (ADMM) to solve the constrained regularization problem.

Let $\bz_i = (z_{i1}, z_{i2}, \dots, z_{iK})\trans$ denote the latent class membership of sample $i$, where $z_{ik} = 1$ if sample $i$ belongs to cluster $k$, and $z_{ik} = 0$ otherwise. Define $\bZ = (\bz_1\trans, \dots, \bz_n\trans)\trans$ as the membership matrix for all $n$ samples. Let $\bM = \{m_{ij}\}\in\real^{n\times p}$ be the data matrix of $n$ independent samples. The complete-data likelihood is given by
\begin{equation*}
  \begin{split}
    L(\bTheta; \bM, \bZ )
    =  \prod_{i=1}^n \prod_{k=1}^K \big\{ \pi_{k} f_{\mathrm{DM}}(\bsm_i;\theta_k, \bbeta_0^{[k]}, \bB^{[k]}) \big \} ^{z_{ik}}.
  \end{split}
\end{equation*}

\subsubsection{EM Algorithm}
In the E-step, the algorithm evaluates the conditional expectations of the latent cluster membership indicators based on the current parameter estimates as
\begin{equation*}
  \widehat{z}_{ik}^{(t+1)} = \mathbb{E} [z_{ik} \mid \bsm_i, \bTheta^{(t)}] = \frac{\pi_k^{(t)} f_{\mathrm{DM}}(\bsm_i;\, \theta_k^{(t)},{\bbeta_0^{[k]}}^{(t)},{\bB^{[k]}}^{(t)})}
  { \sum_{k=1}^K \pi_k^{(t)} f_{\mathrm{DM}}(\bsm_i;\,\theta_k^{(t)}, {\bbeta_0^{[k]}}^{(t)},{\bB^{[k]}}^{(t)})},
\end{equation*}
for the $(t)$-th iteration. Then we get
\begin{equation*}
  \begin{aligned}
    Q(\bTheta \mid \bTheta^{(t)})
    &= \mathbb{E}_{\bZ \mid \bTheta^{(t)},\bM} \bigl[\log L(\bTheta;\,\bM,\bZ) \bigl] = \log L(\bTheta;\,\bM,\widehat{\bZ}^{(t+1)}).
  \end{aligned}
\end{equation*}
To address the non-concavity of the $Q$-function with respect to $\bTheta$, we adopt a majorization-minimization (MM) approach to construct a surrogate function $Q_2(\bTheta \mid \bTheta^{(t)})$ (detailed in Supplement~\ref{supp:E_step}) that minorizes $Q(\bTheta \mid \bTheta^{(t)})$ and is more tractable for optimization \citep{zhou2010mm, zhou2012vs}.

In the M-step, we proceed to solve the following optimization problem,
\begin{equation} \label{eq:E_step}
  \begin{split}
    \min_{\bTheta} \quad - \frac{1}{n} Q_2(\bTheta \mid \bTheta^{(t)}) + \lambda_1 \mathcal{P}_{\gamma}(\bDelta^{[0]}) + \lambda_2 \sum_{k=1}^K \mathcal{P}_{\gamma}(\bDelta^{[k]}),
  \end{split}
\end{equation}
subject to the linear constraints in the objective function~\eqref{eq:optimization}. Let $\bbeta_{0} = ({\bbeta_{0}^{[1]}}\trans, \dots, {\bbeta_{0}^{[K]}}\trans)\trans$ be the stacked vector of intercepts across the $K$ components. Define
\begin{equation*}
  \bdelta = (\bDelta^{[0]}, \bDelta^{[1]}, \dots, \bDelta^{[K]})\trans
  \in \mathbb{R}^{(K+1)p \times q},
\end{equation*}
which collects all coefficient matrices across the $K$ components.
The optimization problem~\eqref{eq:E_step} is separable with respect to $\bpi$, $\btheta$, $\bbeta_0$, and $\bdelta$. There are closed-form solutions for $\bpi$ and $\btheta$. Maximizing $ Q_3(\bbeta_0, \bdelta \mid \bbeta_0^{(t)}, \bdelta^{(t)}) = Q_2(\bpi^{(t+1)},  \btheta^{(t+1)}, \bbeta_0, \bdelta \mid \bpi^{(t+1)},  \btheta^{(t+1)},  \bbeta_0^{(t)}, \bdelta^{(t)})$ with respect to $(\bbeta_0, \bdelta)$ amounts to solving a regularized optimization problem subject to linear constraints. This constrained optimization problem can be efficiently solved using an ADMM algorithm \citep{boyd2011distributed}, which is shown in Supplement~\ref{supp:M_step}.

\subsubsection{E-step}
\label{supp:E_step}

Let $\bTheta^{(t)}$ denote the parameter estimates at the $t$-th iteration. At the $(t+1)$-th iteration, we compute the conditional expectation of the complete-data log-likelihood
\begin{equation*}
\begin{split}
Q(\bTheta \mid \bTheta^{(t)})
&= \mathbb{E}_{\bZ \mid \bTheta^{(t)},\bM} \left\{ \ell(\bTheta;\, \bM, \bZ) \right\} \\
&= \sum_{i=1}^n \sum_{k=1}^K \widehat{z}_{ik}^{(t+1)} \left\{
\log \pi_k
+ \sum_{j=1}^p \sum_{l=0}^{m_{ij} - 1} \log\left(\alpha_{ij}^{[k]} + l\theta_k\right)
- \sum_{l=0}^{M_i - 1} \log\left(1 + l\theta_k\right)
\right\} + \text{C},
\end{split}
\end{equation*}
where $\widehat{z}_{ik}^{(t+1)} = \mathbb{E} [z_{ik} \mid \bsm_i, \bTheta^{(t)}] = \frac{\pi_k^{(t)} f_{\mbox{DM}}(\bsm_i;\, \theta_k^{(t)},{\bbeta_0^{[k]}}^{(t)},{\bB^{[k]}}^{(t)})}
{ \sum_{k=1}^K \pi_k^{(t)} f_{\mbox{DM}}(\bsm_i;\,\theta_k^{(t)}, {\bbeta_0^{[k]}}^{(t)},{\bB^{[k]}}^{(t)})}$ and , and $C$ represents terms constant with respect to $\bTheta$.

Following the reparameterization trick from \cite{zhou2010mm}, an equivalent form of the $Q$-function can be expressed as
\begin{equation*}
\begin{split}
 Q(\bTheta \mid \bTheta^{(t)})
 &= \sum_{i=1}^n \sum_{k=1}^K \widehat{z}_{ik}^{(t+1)} \log \pi_k + \sum_{j=1}^p \sum_{k=1}^K \sum_{l=0}^{m_{ij-1}} \sum_{i=1}^n \widehat{z}_{ik}^{(t+1)} \log(\alpha_{ij}^{[k]} + l\theta_k) \\
 & \quad - \sum_{k=1}^K \sum_{l=0}^{\max(M_i-1)} r_{lk}^{(t+1)} \log(1+l\theta_k)  + \text{C},
\end{split}
\end{equation*}
where $r_{lk}^{(t+1)}=\sum_{i=1}^n \widehat{z}_{ik}^{(t+1)}1\{M_{i} \geqslant l+1\}$.

To address the non-concavity of the $Q$-function with respect to $\bTheta$, we adopt a majorization-minimization (MM) approach, following the strategy proposed \cite{zhou2010mm, zhou2012vs}. Specifically,
by Jensen's inequality, we can minorize
\begin{equation*}
\begin{split}
 \log(\alpha_{ij}^{[k]} + l\theta_k) &\geqslant \frac{\alpha_{ij}^{[k](t)}}{\alpha_{ij}^{[k](t)} + l\theta_k^{(t)}}\log (\frac{\alpha_{ij}^{[k](t)} + l\theta_k^{(t)}}{\alpha_{ij}^{[k](t)}}\alpha_{ij}^{[k]}) + \frac{l\theta_k^{(t)}}{\alpha_{ij}^{[k](t)} + l\theta_k^{(t)}} \log ( \frac{\alpha_{ij}^{[k](t)} + l\theta_k^{(t)}}{l\theta_k^{(t)}}  l\theta_k ) \\
 &= \frac{\alpha_{ij}^{[k](t)}}{\alpha_{ij}^{[k](t)} + l\theta_k^{(t)}}\log \alpha_{ij}^{[k]} + \frac{l\theta_k^{(t)}}{\alpha_{ij}^{[k](t)} + l\theta_k^{(t)}} \log \theta_k  + \text{C}
\end{split}
\end{equation*}
and by the supporting hyperplane property of the convex function, we can minorize
\begin{equation*}
\begin{split}
 -\log(1+l\theta_k) \geqslant -\log(1+l\theta_k^{(t)}) - \frac{l(\theta_k - \theta_k^{(t)})}{1+l\theta_k^{(t)}} = - \frac{l}{1+l\theta_k^{(t)}}\theta_k + \text{C}.
\end{split}
\end{equation*}
Then we construct a surrogate function $Q_2(\bTheta \mid \bTheta^{(t)})$ that minorizes $Q(\bTheta \mid \bTheta^{(t)})$ and is more tractable for optimization as
\begin{equation*}
\begin{split}
Q_2(\bTheta \mid \bTheta^{(t)})
&= \sum_{i=1}^n \sum_{k=1}^K \widehat z_{ik}^{(t+1)} \log \pi_k +
\sum_{j=1}^p \sum_{k=1}^K  \sum_{i=1}^n  \widehat{z}_{ik}^{(t+1)} S_{ijk}^{(t+1)} \log(\alpha_{ij}^{[k]}) \\
& \quad + \sum_{j=1}^p \sum_{k=1}^K \sum_{i=1}^n  \widehat{z}_{ik}^{(t+1)} T_{ijk}^{(t+1)} \log \theta_k - \sum_{k=1}^K R_k^{(t+1)} \theta_k + \text{C}\\
&= \sum_{i=1}^n \sum_{k=1}^K \widehat{z}_{ik}^{(t+1)} \log \pi_k +
\sum_{j=1}^p \sum_{k=1}^K  \sum_{i=1}^n \widehat{z}_{ik}^{(t+1)} S_{ijk}^{(t+1)} (\beta_{0j}^{[k]} + \bx_i\trans\bbeta_{j}^{[k]}) \\
& \quad -  \sum_{k=1}^K  \sum_{i=1}^n  \widehat{z}_{ik}^{(t+1)} \bigl(\sum_{j=1}^p S_{ijk}^{(t+1)}\bigr) \log \bigl(\sum_{j=1}^p \exp(\beta_{0j}^{[k]} + \bx_i\trans\bbeta_{j}^{[k]})\bigr)  \\
& \quad +  \sum_{j=1}^p \sum_{k=1}^K  \sum_{i=1}^n  \widehat z_{ik}^{(t+1)} T_{ijk}^{(t+1)} \log \theta_k - \sum_{k=1}^K R_k^{(t+1)} \theta_k + \text{C},
\end{split}
\end{equation*}
where
$s_{ijkl}^{(t+1)} = {{\alpha_{ij}^{[k](t)}}}/{\bigl({\alpha_{ij}^{[k] (t)}}+ l\theta_k^{(t)}\bigr)}$,
$S_{ijk}^{(t+1)} = \sum_{l=0}^{m_{ij}-1} s_{ijkl}^{(t+1)}$,
$T_{ijk}^{(t+1)} = \sum_{l=0}^{m_{ij}-1} (1-s_{ijkl}^{(t+1)}) = m_{ij} - S_{ijk}^{(t+1)}$, $R_k^{(t+1)} = \sum_{l=0}^{\max(M_i - 1)} r_{lk}^{(t+1)} {l}/{\bigl(1+l \theta_k^{(t)}\bigr)}$.

We proceed to solve the following optimization problem in M-step.
\begin{equation}
\begin{split}
\min_{\bTheta} \quad - \frac{1}{n} Q_2(\bTheta \mid \bTheta^{(t)}) + \lambda_1 \mathcal{P}_{\gamma}(\bDelta^{[0]}) + \lambda_2 \sum_{k=1}^K \mathcal{P}_{\gamma}(\bDelta^{[k]}),
\end{split}
\end{equation}
subject to the linear constraints in \eqref{eq:optimization}.

\subsubsection{M-step}
\label{supp:M_step}

The objective of the M-step is to update $\bTheta^{(t+1)}$ by minimizing the expression in \eqref{eq:E_step}. Let $\bbeta_{0} = ({\bbeta_{(0)}^{[1]}}\trans, \dots, {\bbeta_{(0)}^{[K]}}\trans)$ denote the stacked vector of intercepts across the $K$ components, and let  $\bdelta = (\bDelta^{[0]}, \bDelta^{[1]}, \dots, \bDelta^{[K]})\trans \in \mathbb{R}^{(K+1)p \times q}$ collect all coefficient matrices across the $K$ components.
The optimization problem is separable with respect to $\bpi$, $\btheta$, $\bbeta_0$, and $\bdelta$.

The update for $\bpi$ is obtained by solving the following constrained optimization problem: \begin{equation*}
\begin{split}
\bpi^{(t+1)} = \arg\min_{\bpi}
\left\{ - \sum_{i=1}^n \sum_{k=1}^K \widehat{z}_{ik}^{(t+1)} \log \pi_k
\right\}, \quad \text{s.t. } \sum_{k=1}^K \pi_k = 1.
\end{split}
\end{equation*}
The closed-form solution is given by $\pi_k^{(t+1)} = \frac{1}{n} \sum_{i=1}^n \widehat{z}_{ik}^{(t+1)}$, for $k = 1, \dots, K$.

The update for $\btheta$ is obtained by solving the following optimization problem
\begin{equation*}
\begin{split}
\btheta^{(t+1)} = \arg\min_{\btheta} \left\{ - \sum_{j=1}^p \sum_{k=1}^K \sum_{i=1}^n \widehat{z}_{ik}^{(t+1)} T_{ijk}^{(t+1)} \log \theta_k + \sum_{k=1}^K R_k^{(t+1)} \theta_k \right\}, \quad \text{s.t. } \theta_k > 0.
\end{split}
\end{equation*}
The resulting closed-form update is given by $\theta_k^{(t+1)}= \bigl(\sum_{j=1}^p \sum_{i=1}^n  \widehat z_{ik}^{(t+1)} T_{ijk}^{(t+1)}\bigr)/{R_k^{(t+1)}}$, for $k=1, \dots, K$.

The update for ${\bbeta_0, \bdelta}$ is obtained by solving the following regularized optimization problem
\begin{equation}  \label{eq:obj_delta}
 \bbeta_0^{(t+1)}, \bdelta^{(t+1)}
 = \arg\min_{\bbeta_0, \bdelta} - \frac{1}{n} Q_3(\bbeta_0, \bdelta \mid \bbeta_0^{(t)}, \bdelta^{(t)})
  + \lambda_1 \mathcal{P}_{\gamma}(\bDelta^{[0]}) + \lambda_2 \sum_{k=1}^K \mathcal{P}_{\gamma}(\bDelta^{[k]}),
\end{equation}
subject to the linear constraints in \eqref{eq:optimization} and
\begin{equation*}
\begin{split}
Q_3(\bbeta_0, \bdelta \mid \bbeta_0^{(t)}, \bdelta^{(t)})
&=  \sum_{j=1}^p \sum_{k=1}^K  \sum_{i=1}^n  \widehat{z}_{ik}^{(t+1)} S_{ijk}^{(t+1)} (\beta_{0j}^{[k]} + \bx_i\trans\bbeta_{j}^{[k]}) \\
& \qquad - \sum_{k=1}^K  \sum_{i=1}^n  \widehat{z}_{ik}^{(t+1)} \bigl(\sum_{j=1}^p S_{ijk}^{(t+1)}\bigr) \log
\bigl(\sum_{j=1}^p \exp(\beta_{0j}^{[k]} + \bx_i\trans\bbeta_{j}^{[k]})\bigr),
\end{split}
\end{equation*}
where the component-specific coefficients $\bbeta_j^{[k]}$ are functions of $\bdelta$ through the reparameterization given in \eqref{eq:reparam}.

The constrained optimization problem in \eqref{eq:obj_delta} can be efficiently solved using the Alternating Direction Method of Multipliers (ADMM) algorithm \citep{boyd2011distributed}. We introduce the augmented parameter $\widetilde\bdelta = (\bdelta_0, \bdelta) \in \mathbb{R}^{p(K+1) \times (q+1)}$, where $\bdelta_0 = \big( {\bdelta_0^{[0]}}\trans, {\bdelta_0^{[1]}}\trans, \dots, {\bdelta_0^{[K]}}\trans \big)\trans$, and each $\bdelta_0^{[k]}$ for $k = 0, 1, \dots, K$ is defined such that $\bbeta_0^{[k]} = \bdelta_0^{[0]} + \bdelta_0^{[k]}$. We also introduce the dual variable $\widetilde\bb = \widetilde\bdelta$, and denote $\bb^{[k]} = \bDelta^{[k]} \in \mathbb{R}^{p \times q}$ for $k = 0, 1, \dots, K$, with $\bb_{(l)}^{[k]} = \bdelta_{(l)}^{[k]} \in \mathbb{R}^{p}$ for $l = 1, \dots, q$.

With this reparameterization, the optimization problem in \eqref{eq:obj_delta} can be equivalently reformulated to facilitate the application of ADMM as follows
\begin{equation} \label{eq:M_step}
\begin{split}
& \min_{\widetilde\bdelta,\widetilde\bb} \quad - \frac{1}{n} Q_3(\widetilde\bdelta \mid \widetilde\bdelta^{(t)}) + \lambda_1 \mathcal{P}_{\gamma}(\bb^{[0]}) + \lambda_2 \sum_{k=1}^K \mathcal{P}_{\gamma}(\bb^{[k]}) \\
& \quad \qquad s.t. \ \ \widetilde\bdelta = \widetilde\bb, \\
& \quad \qquad \quad \quad \bDelta^{[0]}\1=\0,\  \bDelta^{[k]}\1 = \0, \ k = 1, \ldots, K, \\
& \quad \qquad \qquad \sum_{k=1}^K \bDelta^{[k]} = \0, \\
& \quad \qquad \qquad  {\bbeta_0^{[k]}}\trans \1 = 0, \ k = 0, \ldots, K.
\end{split}
\end{equation}
The first two linear constraints can be combined and expressed in matrix form as
\begin{equation*}
\begin{split}
\begin{pmatrix}
 \bI_{(K+1)p} \\
(\1_K,\bI_{K}) \otimes \1_p\trans
\end{pmatrix}
\widetilde\bdelta +
\begin{pmatrix}
-\bI_{(K+1)p}  \\
\boldsymbol{0}_{K\times(K+1)p}
\end{pmatrix}
\widetilde\bb
=
\boldsymbol{0}_{(pK+p+K)\times (q+1)},
\end{split}
\end{equation*}
which we denote compactly as $\bA_1 \widetilde\bdelta + \bA_2 \widetilde\bb  = \0$.

The optimization problem in \eqref{eq:M_step} is solved via an ADMM algorithm, which iteratively updates the primal and dual variables until convergence. Specifically, at the $\{s+1\}$-th iteration, the updates are given by the following steps
\begin{subequations}
\label{eq:ADMM}
\begin{align}
& \widetilde\bdelta^{\{s+1\}}  = \arg\min_{\widetilde\bdelta} - \frac{1}{n} Q_3(\widetilde\bdelta \mid \widetilde\bdelta^{\{s\}}) + \frac{\rho}{2} \|\bA_1  \widetilde\bdelta + \bA_2 \widetilde\bb^{\{s\}}  + \u^{\{s\}}  \|_F^2 \label{eq:ADMM_delta} \\
& \qquad \qquad \mbox{s.t. } \sum_{k=1}^K \bdelta^{[k]} = \0_{p\times (q+1)} , \notag \\
& \quad \qquad \qquad \quad  {\bbeta_0^{[k]}}\trans \1 = 0, \ k = 0, \ldots, K. \notag \\
& \widetilde\bb^{\{s+1\}} = \arg\min_{\widetilde\bb} \   \frac{\rho}{2} \|\bA_1  \widetilde\bdelta^{\{s+1\}} + \bA_2 \widetilde\bb + \u^{\{s\}}  \|_F^2 + \lambda_1 \mathcal{P}_{\gamma}(\bb^{[0]}) + \lambda_2 \sum_{k=1}^K \mathcal{P}_{\gamma}(\bb^{[k]}) \label{eq:ADMM_b} \\
& \u^{\{s+1\}}  = \u^{\{s\}}  + \bA_1  \widetilde\bdelta^{\{s+1\}} + \bA_2 \widetilde\bb^{\{s+1\}},   \notag
\end{align}
\end{subequations}
where $\| \cdot \|_F$ is the Frobenius norm.

\subsubsection{Details of ADMM Algorithm}

We could similarly denote $\widetilde\bbeta = (\bbeta_0, \bbeta)\in \real^{pK \times (q+1)}$, where $\bbeta = (\bbeta^{[1]}, \dots, \bbeta^{[K]})\trans \in \real^{pK \times q}$ and $\bbeta_0 = ({\bbeta_0^{[1]}}\trans, \dots, {\bbeta_0^{[K]}}\trans)\trans \in \real^{pK}$. Then $\widetilde\bdelta$ can be written as a linear function of $\widetilde\bbeta$ as
\begin{equation*}
\begin{split}
\widetilde\bdelta = \bH \widetilde\bbeta, \ \bH =
\begin{pmatrix}
\frac{1}{K} ( \1_{K}\trans \otimes \bI_p )\\
\boldsymbol{I}_{pK} - \frac{1}{K} ( \bJ_K \otimes \bI_p)
\end{pmatrix} \in \real^{p(K+1)\times pK},
\end{split}
\end{equation*}
where $\bJ_K$ is $K\times K$ matrix of ones.
Thus, solving for $\widetilde\bdelta$ in \eqref{eq:ADMM_delta} is equivalent to solving for $\widetilde\bbeta$ in the following optimization problem
\begin{equation}
\begin{split} \label{eq:admm_beta}
& \widetilde\bbeta^{\{s+1\}}
= \arg\min_{\widetilde\bbeta} - \frac{1}{n} \Big\{ \sum_{j=1}^p \sum_{k=1}^K  \sum_{i=1}^n  \widehat z_{ik}^{(t+1)} S_{ijk}^{(t+1)} (\beta^{[k]}_{0j} + \bx_i\trans\bbeta^{[k]}_j) \\
& \qquad \qquad \qquad \qquad -  \sum_{k=1}^K  \sum_{i=1}^n  \widehat z_{ik}^{(t+1)} \bigl(\sum_{j=1}^p S_{ijk}^{(t+1)}\bigr) \log \bigl(\sum_{j=1}^p \exp(\beta^{[k]}_{0j} + \bx_i\trans\bbeta^{[k]}_j)\bigr) \Big\} \\
& \qquad \qquad \qquad \qquad + \frac{\rho}{2} \|\bA_1 \bH \widetilde\bbeta + \bA_2 \widetilde\bb^{\{s\}}  + \u^{\{s\}} \|_F^2
\end{split}
\end{equation}
The unconstrained problem~\eqref{eq:admm_beta} is solved with a limited-memory BFGS (L-BFGS) algorithm \citep{liu1989limited}, implemented in C++ via \texttt{Rcpp}.

Next, we formulate the optimization problem with respect to $\widetilde\bb$ in \eqref{eq:ADMM_b} as
\begin{equation*}
\begin{split}
\widetilde\bb^{\{s+1\}}
& = \arg\min_{\widetilde\bb} \frac{\rho}{2} \|\bA_1  \widetilde\bdelta^{\{s+1\}} + \bA_2 \widetilde\bb + \u^{\{s\}} \|_F^2  + \lambda_1 \mathcal{P}_{\gamma}(\bb^{[0]}) + \lambda_2 \sum_{k=1}^K \mathcal{P}_{\gamma}(\bb^{[k]}) \\
& = \arg\min_{\widetilde\bb}   \frac{1}{2} \sum_{l=1}^q \|\bA_1  \widetilde\bdelta^{\{s+1\}}_{(l)} + \bA_2 \widetilde\bb_{(l)} + \u^{\{s\}}_{(l)} \|^2  + \frac{\lambda_1}{\rho} \sum_{l=1}^q \mathcal{P}_{\gamma}(\bb^{[0]}_{(l)}) + \frac{\lambda_2}{\rho} \sum_{k=1}^K\sum_{l=1}^{q} \mathcal{P}_{\gamma}(\bb^{[k]}_{(l)}),
\end{split}
\end{equation*}
where $\widetilde\bdelta_{(l)}^{\{s+1\}}$, $\widetilde\bb_{(l)}$, and $\u_{(l)}^{\{s\}}$ denote the $(l+1)$-th columns of the matrices $\widetilde\bdelta^{\{s+1\}}$, $\widetilde\bb$, and $\u^{\{s\}}$, respectively, for $l = 1, \dots, q$. This optimization problem is separable across the intercept term $\bb_0$ and each column $\widetilde\bb_{(l)}$.
\begin{itemize}
\item Intercept term ($\bb_0$) The optimization with respect to the
  intercept component $\bb_0$ reduces to
  \begin{equation*}
    \bb_0^{\{s+1\}} = \arg\min_{\bb_0} \sum_{k=0}^K \frac{1}{2} \left| \bb_0^{[k]} - \left( {\bdelta_0^{[k]}}^{\{s+1\}} + \widehat{\u}_0^{[k]\{s\}} \right) \right|^2, \end{equation*} which is separable across $k$, yielding the closed-form update
  $\bb_0^{[k]\{s+1\}} = \bdelta_0^{[k]\{s+1\}} + \widehat{\u}_0^{[k]\{s\}}, \quad k = 0, \dots, K$.
\item Coefficient vectors ($\bb_{(l)}$, $l = 1, \dots, q$) The update
  for each $\bb_{(l)}$ is given by
  \begin{equation*}
    \begin{split}
      \bb_{(l)}^{\{s+1\}} & = \arg\min_{\bb_{(l)}}   \frac{1}{2}  \|\bA_1  \widetilde\bdelta^{\{s+1\}}_{(l)} + \bA_2 \bb_{(l)} + \u^{\{s\}}_{(l)} \|^2 + \frac{\lambda_1}{\rho}  \mathcal{P}_{\gamma}(\bb^{[0]}_{(l)})  + \frac{\lambda_2}{\rho}    \sum_{k=1}^K \mathcal{P}_{\gamma}(\bb^{[k]}_{(l)})  \\
                          &=  \arg\min_{\bb_{(l)}}   \frac{1}{2}  \| \bb_{(l)} - (\widetilde\bdelta^{\{s+1\}}_{(l)}    + \widehat{\u}^{\{s\}}_{(l)} ) \|^2 + \frac{\lambda_1}{\rho}  \mathcal{P}_{\gamma}(\bb^{[0]}_{(l)}) + \frac{\lambda_2}{\rho}  \sum_{k=1}^K \mathcal{P}_{\gamma}(\bb^{[k]}_{(l)}) ,
    \end{split}
  \end{equation*}
  where $\widehat{\u}^{\{s\}}_{(l)}$ is the sub-vector containing the
  first $pK$ elements of vector $\u^{\{s\}}_{(l)}$. This objective is
  separable across the group-specific components $\bb_{(l)}^{[k]}$ for
  $k = 0, \dots, K$. When $\mathcal{P}_{\gamma}(\cdot)$ is the adaptive group
  $\ell_2$ penalty, i.e.,
  $\mathcal{P}_{\gamma}(\bb_{(l)}^{[k]}) = w_{(l)}^{[k]} \|\bb_{(l)}^{[k]}\|$, this subproblem admits a closed-form solution
  via group soft-thresholding:
    $${\bb_{(l)}^{[k]}}^{\{s+1\}} =  S({\bdelta_{(l)}^{[k]}}^{\{s+1\}}    + (\widehat{\u}_{(l)}^{[k]})^{\{s\}} ,\; \frac{\lambda_2 w_{(l)}^{[k]}}{\rho}),$$
    for $k = 1,\dots, K$ and
    $${\bb_{(l)}^{[0]}}^{\{s+1\}} =  S({\bdelta_{(l)}^{[0]}}^{\{s+1\}}    + (\widehat{\u}_{(l)}^{[0]})^{\{s\}} ,\; \frac{\lambda_1 w_{(l)}^{[0]}}{\rho})$$
    for $k = 0$, where
    $S(\bz, \lambda) = \left( 1 - \frac{\lambda}{|\bz|} \right)_+ \bz$
    is the group soft-thresholding operator.
  \end{itemize}

\subsubsection{Algorithm}\label{supp:algorithm}

The computational procedure is outlined in Algorithm~\ref{alg:EM} below.
\begin{algorithm}
\caption{EM algorithm for the DMMR model}\label{alg:EM}
\begin{algorithmic}
\State \textbf{Initialization}: $\bTheta^{(0)}=\{\bpi^{(0)}, \btheta^{(0)}, \bbeta_{0}^{(0)}, \bdelta^{(0)}\}$; tolerance for stopping condition $\mathrm{tol}_{\mathrm{EM}}$ (e.g., $10^{-4}$), maximum number of iterations $\mathrm{maxiter}_{\mathrm{EM}}$ (e.g., $100$). Set iteration number $t \gets 0$.
\Repeat \ when $t < \mathrm{maxiter}_{\mathrm{EM}}$
\State \textbf{E-step}:
\State \qquad $\widehat\bZ^{(t+1)} \gets \arg\max_{\bZ} \bigl\{\log L(\bTheta^{(t)};\, \bM, \bZ) \bigr\}$ with closed form.
\State \qquad Update the surrogate $Q_2$ function.
\State \textbf{M-step}:
\State \qquad $\bpi^{(t+1)} \gets \arg\max_{\bpi} Q_2(\bpi, \btheta^{(t)}, \bbeta_{0}^{(t)}, \bdelta^{(t)}\mid \bpi^{(t)}, \btheta^{(t)}, \bbeta_{0}^{(t)}, \bdelta^{(t)})$ with closed form.
\State \qquad $\btheta^{(t+1)} \gets \arg\max_{\btheta} Q_2(\bpi^{(t+1)}, \btheta, \bbeta_{0}^{(t)}, \bdelta^{(t)}\mid \bpi^{(t+1)}, \btheta^{(t)}, \bbeta_{0}^{(t)}, \bdelta^{(t)})$ with closed form.
\State \qquad $(\bbeta_{0}^{(t+1)}$, $\bdelta^{(t+1)}) \gets
\arg\max_{\bbeta_{0}, \bdelta}
\bigl\{ Q_3(\bbeta_0, \bdelta \mid \bbeta_0^{(t)}, \bdelta^{(t)}) + \lambda_1 \mathcal{P}_{\gamma}(\bDelta^{[0]}) + \lambda_2 \sum_{k=1}^K \mathcal{P}_{\gamma}(\bDelta^{[k]})\bigr\}$ with the ADMM algorithm in Supplement~\ref{supp:M_step}.
\State $t \gets t+1$
\Until{ $\|\mathrm{vec}(\bTheta^{(t+1)}) - \mathrm{vec}(\bTheta^{(t)})\|/(\|  \mathrm{vec}(\bTheta^{(t)}) \|+ 10^{-14}) \leqslant  \mathrm{tol}_{\mathrm{EM}}$. }
\end{algorithmic}
\end{algorithm}

\subsection{Solution Path \& Model Selection}
\label{supp:selection}

We employ the following algorithm to estimate $\lambda_{\max}$ for which the estimated $\bdelta$ becomes a zero-matrix, indicating no covariate effects.

\begin{algorithm}[!htb]
\caption{Procedure to find $\lambda_{\max}$ of the tuning parameter sequence}
\begin{algorithmic}
\State \textbf{Initialization}: initial testing value of $\lambda_{\max}$: $\lambda = 1$, lower and upper bounds of $\lambda_{\max}$: $\lambda_l = 0, \lambda_u = 100$, number of EM iterations: $a = 10$, number of bisection steps for the search procedure: $b = 10$.
\State $i \gets 0$
\Repeat
\State Run the EM algorithm with tuning parameter $\lambda$ for $a$ iterations.
    \If{ $\bdelta^{[k]} = \boldsymbol{0}$ for all $k = 1, \dots, K$ during iterations}
        \State Set $\lambda_u \gets \lambda$
    \Else
        \State Set $\lambda_l \gets \lambda$
    \EndIf
\State Update $\lambda \gets (\lambda_l + \lambda_u)/2$
\State $i \gets i+1$
\Until{ $i = b$ }
\State $\lambda_{\max} \gets \lambda$.
\end{algorithmic}
\end{algorithm}

To select the optimal number of mixture components and tuning parameters, there are several model selection criteria, including the Akaike Information Criterion (AIC), Bayesian Information Criterion (BIC), and Generalized Information Criterion (GIC). These criteria are evaluated over a grid of candidate values for $\{K, \lambda\}$ and rely on the computation of model complexity, which is quantified by the degrees of freedom (df) in the proposed DMMR model with heterogeneous pursuit.

We define the number of active covariates for the common-effect component $\bDelta^{[0]}$ and for each cluster-specific effect $\bDelta^{[k]}$. Due to the constraint $\sum_{k=1}^K \bDelta^{[k]} =\0$, the degrees of freedom are reduced accordingly. The resulting total degrees of freedom associated with $\bdelta$ as follows:
\begin{equation*}
\begin{aligned}
    S_k &= \{\text{indices of non-zero rows in $\bDelta^{[k]}$}\} \\
    &= \{l: \bdelta_{(l)}^{[k]}\ne \0\},\quad k = 0,1,\dots,K;\\
    s_k &=\# S_k,\quad k = 0,1,\dots,K;\\
    s_{c} &= \# \bigl(\cup_{k=1}^K S_k\bigr);\\
    \mathrm{df}(\bdelta) &= \bigr(\sum_{k=0}^q s_k - s_c\bigr) \times (p-1)
\end{aligned}
\end{equation*}

The total degrees of freedom for the DMMR model with heterogeneity pursuit, incorporating the contributions from the mixture proportions, over-dispersion parameters, intercepts, and covariate effects, is given by
\begin{equation*}
\begin{aligned}
\mathrm{df} &= \mathrm{df}(\bpi) + \mathrm{df}(\btheta) + \mathrm{df}(\bbeta_{0}) + \mathrm{df}(\bdelta)\\
&= 2K-1 + \bigr(K + \sum_{k=0}^K s_k - s_c\bigr) \times (p-1).
\end{aligned}
\end{equation*}

Based on this df estimate, the selection criteria are defined as
\begin{equation*}
    -2\times \log L(\bTheta;\, \bM) + a_n \times \mathrm{df},
\end{equation*}
where the term $a_n$ depends on the criterion: $a_n=\log(n)$ for BIC, and $a_n=\log \bigl(\log (n)\bigr) \times \log(\max \{n, \mathrm{df}_{\max}\})$ for GIC, with $\mathrm{df}_{\max} = 2K - 1 + K(q+1)(p-1)$.

For the adaptively penalized estimator (DMMR-AP), the initial estimates used to construct the adaptive weights are obtained from the DMMR solution path evaluated at $\lambda = 0.001$, which serves as a near-unpenalized solution. Because fitting the model at exactly $\lambda = 0$ can be numerically unstable, especially in higher-dimensional settings, this approach leverages the warm-start structure of the solution path to produce a stable initial estimator.
With this initial estimator $\widehat{\bdelta}_{(l)}^{[k](0)}$, the adaptive weights are computed as $w^{[k]}_{(l)} = 1/(\|\widehat{\bdelta}_{(l)}^{[k](0)}\| + 10^{-6})$ with $\gamma=1$. The solution path is then recomputed using these adaptive weights, and model selection proceeds as described below.

Throughout the paper, unless otherwise noted, for the model-based approaches (DMM, DMMR-NP, DMMR-P, and DMMR-AP), the number of clusters was selected using BIC. For K-Means and HC, $K$ was determined by maximizing the average silhouette width, which measures cluster cohesion and separation. Since the silhouette width is not defined for $K=1$, we adopted a threshold-based rule: when the silhouette widths for both $K=2$ and $K=3$ were below $0.15$, we selected $K=1$ to indicate no meaningful clustering structure.

\FloatBarrier

\clearpage
\section{Proof}
\label{supp:proof}

\begin{assumption}[Regularity condition]
\label{supp:assumption}
Let $\bTheta = \{\bpi, \btheta, \bbeta_{0}^{[1]},\ldots,\bbeta_{0}^{[K]}, \bB^{[1]},\ldots, \bB^{[K]}\}$ to collect all unknown parameters, and the parameter space is given by
\begin{equation*}
  \Omega=\Pi \times \mathbb R_{+} \times \Xi,
\end{equation*}
where $\Pi=\{(p_1,\dots, p_K)\trans: 0<p_k<1,\,\sum_{k=1}^K p_k=1\}$. Use $\bTheta_j$ to denote the $k$-th entry of vectorized $\bTheta$.

We assume similar conditions as in work \citep{fan2001variable}, but with necessary modification due to natural constraints on the parameters for the model to be identifiable. Let $\bV_i = (\bsm_i,\bx_i)$ be the $i$th observation for $i = 1,\dots,n$.
\begin{enumerate}
\item The observations $\bV_i$ are independent and identically distributed with probability density function $f(\bV;\bTheta)$ with respect to some measure $\mu$. $f(\bV;\bTheta)$ has a common support and the model is identifiable. Furthermore, the first and second logarithmic derivatives of $f$ satisfying the equations
  \begin{equation*}
    \mathrm{E}_{\bTheta}\biggr[\frac{\partial \log f(\bV;\bTheta)}{\partial \bTheta}\biggr] = \0
  \end{equation*} and
  \begin{align*}
    \bI_{jk}(\bTheta)
    &=
      \mathrm{E}_{\bTheta}\biggr[\frac{\partial \log f(\bV;\bTheta)}{\partial \bTheta_j } \frac{\partial \log f(\bV;\bTheta)}{\partial \bTheta_k}\biggr] \\
    &=
      \mathrm{E}_{\bTheta}\biggr[-\frac{\partial^2 \log f(\bV;\bTheta)}{\partial \bTheta_j \bTheta_k}\biggr]
  \end{align*}
\item The Fisher information matrix
  \begin{equation*}
    \bI(\bTheta)=
    \mathrm{E}_{\bTheta}\biggr[\biggl(\frac{\partial \log f(\bV;\bTheta)}{\partial \bTheta }\biggr) \biggl(\frac{\partial \log f(\bV;\bTheta)}{\partial \bTheta}\biggr)\trans\biggr]
  \end{equation*}
  is finite and positive definite at the true parameter vector $\bTheta = \bTheta^\star$ with respect to the constraints.
\item There exists an open subset $\omega$ of $\Omega$ that contains the true parameter $\bTheta^\star$ such that for almost all $\bV$, the density function $f(\bV;\bTheta)$ admits all third derivatives. Furthermore, there exists function $M_{jkl}(\cdot)$ such that
  \begin{equation*}
    \biggr|\frac{\partial^3}{\partial \bTheta_j\bTheta_k\bTheta_l }
    \log f(\bV;\bTheta)\biggr| \le
    M_{jkl}(\bTheta) \text{ for all } \bTheta \in \omega.
  \end{equation*}
\end{enumerate}
\end{assumption}

\begin{proof}[Proof of Theorem~\ref{thm:group_lasso_ini}]

Let
\begin{align*}
    T_g &=
    \begin{pmatrix}
    \0_{K(p+2)}\trans,&
    \frac{1}{K}(\bI_p\otimes \be_{l,q}\trans ) (\1_K\trans \otimes \bI_{qp})
    \end{pmatrix},
    \quad \text{ for }g \in \{1,\dots,q\};\\
    T_g &=
    \begin{pmatrix}
    \0_{K(p+2)}\trans,&
    (\bI_p\otimes \be_{l,q}\trans )\left((\be_{k,K}\trans-\frac{1}{K}\1_K\trans) \otimes \bI_{qp}\right)
    \end{pmatrix},
    \quad \text{ for } g \in \{q+1,\dots, (K+1)q\},
\end{align*}
where $\be_{l,q}$ is the length-$q$ vector of zeros with a $1$ in the $l$-th entry, and $\be_{k,K}$ is defined similarly, then $T_g\bTheta = \delta_{(l)}^{[k]}$ where $g = qk + l$ corresponds to the quotient $k$ and remainder $l$ of the division by $q$.

Let
\begin{align*}
    \bD =
    \begin{pmatrix}
    \1_K\trans,&\0,&\0\\
    \0, & \bI_K\otimes\1_p\trans, &\0\\
    \0, & \0, &(\bI_K\otimes\1_p\trans)\otimes \bI_q\\
    \end{pmatrix}
\end{align*}
and $\bd = (1,\0_{K(q+1)}\trans)\trans$, so that $\bD\bTheta = \bd$ collects all linear constraints.

Define
\begin{equation*}
  l_\lambda (\bTheta) = l(\bTheta)-n\lambda \sum_{g=1}^G \|\bT_g \bTheta\|_2.
\end{equation*}
Let $\bu=\sqrt n (\bTheta-\bTheta^\star)$. Define
\begin{equation*}
  D_n(\bu)=l_\lambda(\bTheta^\star+n^{-1/2}\bu)-l_\lambda(\bTheta^\star).
\end{equation*}

In order to show $\widehat\bTheta$ is root $n$-consistent, we need to show for any $\varepsilon>0$, there exists a sufficiently large constant $c$ such that
\begin{equation*}
  P\bigg(\sup_{\|\bu\|=c,\; \bD\bu=\0} D_n(\bu) <0\bigg)\ge 1-\varepsilon
\end{equation*}

\begin{equation*}
\begin{split}
D_n(\bu)
    &=l(\bTheta^\star+n^{-1/2}\bu)-l(\bTheta^\star)\\
    &-n\lambda\bigg\{
    \sum_{g=1}^G \|\bT_g (\bTheta^\star+n^{-1/2} \bu)\|_2
    -\sum_{g=1}^G \|\bT_g  \bTheta^\star\|_2
    \bigg\}\\
    &\le l(\bTheta^\star+n^{-1/2}\bu)-l(\bTheta^\star)\\
    &-\underbrace{n\lambda\bigg\{
    \sum_{g=1}^G \big[\|\bT_g (\bTheta^\star + n^{-1/2} \bu)\|_2
    - \|\bT_g  \bTheta^\star\|_2
    \big]I(\|\bT_g  \bTheta^\star\|_2\ne 0)
    \bigg\} }_{G_n(\bu)}.
\end{split}
\end{equation*}

\begin{equation*}
\begin{split}
|G_n(\bu)|
    &=n\lambda\bigg\{
    \sum_{g=1}^G \big[\|\bT_g (\bTheta^\star + n^{-1/2} \bu)\|_2
    - \|\bT_g  \bTheta^\star \|_2
    \big]I(\|\bT_g  \bTheta^\star\|_2\ne 0)
    \bigg\}\\
    &\le \lambda \sqrt n \sum_{g=1}^G \|\bT_g  \bu\|_2
    I(\|\bT_g  \bTheta^\star\|_2\ne 0) \\
    &\sim O_p(\lambda\sqrt  n \|\bu\|_2)
\end{split}
\end{equation*}
As long as $\lambda =O(n^{-1/2})$, $D_n(\bu)=-\frac{1}{2}\bu\trans \bI(\bTheta^\star)\bu+R_n(\bu)$ and $R_n(\bu)=o_p(\|\bu\|^2)$. With similar arguments as before, $\widehat \bTheta$ achieves $\sqrt n$-consistency.
\end{proof}

\begin{proof}[Proof of Corollary~\ref{thm:group_lasso_adp}]

Define
\begin{equation*}
  l_\lambda^\gamma (\bTheta) = l(\bTheta)-n \sum_{g=1}^G \lambda_g\|\bT_g \bTheta\|_2,
\end{equation*}
where $\lambda_g=\lambda \|\bT_g \widehat \bTheta^{(0)}\|_2^{-\gamma}$, and $\widehat\Delta^{(0)}$ is any non-adaptive estimator with $\sqrt n$-consistency.
For notation convenience, let $\bT_g \bTheta^\star\ne\0$ for $g=1,\dots,g_0$ while $\bT_g \bTheta^\star=\0$ for $g>g_0$ without loss of generality. Define $a_n=\max\{ \lambda_g, g\le g_0\}$ and $b_n = \min\{\lambda_g, g > g_0\}$, then $a_n=O(\lambda),b_n=O(\lambda n^{\gamma/2})$.
Beside,
\begin{equation*}
\begin{split}
  D_n^\gamma(\bu)
  &=l(\bTheta^\star + n^{-1/2}\bu)-l(\bTheta^\star)\\
  &-n\sum_{g=1}^G \lambda_g
  \bigg\{
  \|\bT_g (\bTheta^\star + n^{-1/2} \bu)\|_2  -\sum_{g=1}^G \|\bT_g  \bTheta^\star\|_2
  \bigg\}\\
  &\le \frac{1}{\sqrt n}l^\prime(\bTheta^\star)\trans \bu
            -\frac{1}{2}\bu\trans \bI(\bTheta^\star)\bu
            +3g_0 a \sqrt n \|\bu\|_2
            +o_p(\|\bu\|^2).
\end{split}
\end{equation*}
As long as $a\sqrt n \le O(1)$, i.e., $\lambda\sqrt n \le O(1)$, $\widehat \bTheta$ achieves $\sqrt n$-consistency.

The KKT condition now is
\begin{equation*}
  \begin{cases}
    &-n\bI(\bTheta^\star) (\widehat \bTheta-\bTheta^\star)+l^\prime(\bTheta^\star)+R_n^\prime(\widehat\bTheta-\bTheta)
      +n\partial P_\lambda(\widehat{\bTheta})
      +\bD\trans \widehat\bu  \ni \0,\\
    &\bD\widehat{\bTheta} = \bd,
  \end{cases}
\end{equation*}
where $\partial P_\lambda(\cdot)$ is the set of sub-gradients of the full penalty term $P_\lambda$, induced by the componentwise penalty $\mathcal{P}_{\gamma}(\cdot)$, and $R_n$ is the remainder term in Taylor expansion.

Denote $\bM=\bI(\bTheta^\star)^{-1}-\bI(\bTheta^\star)^{-1}\bD (\bD \bI(\bTheta^\star)\bD\trans)^{+}\bD\trans \bI(\bTheta^\star)^{-1}$, then the solution can be written as
\begin{equation*}
  \widehat\bTheta = \bTheta^\star - \frac{1}{n}\bM l^\prime(\bTheta^\star)-\frac{1}{n}\bM R_n^\prime(\widehat\bTheta-\bTheta^\star)-\bM \widetilde {\partial P}_\lambda(\bTheta),
\end{equation*}
where $\widetilde {\partial P}_\lambda(\cdot)$ is one particular element in all sub-gradients.

It follows that
\begin{equation}
\label{eq:sign_group}
\begin{aligned}
\sqrt n\bT_g\widehat\bTheta
    &=\sqrt n\bT_g\widehat\bTheta^\star
    - \bT_g \bM \frac{l^\prime(\bTheta^\star)}{\sqrt n}-\frac{1}{\sqrt n} \bT_g\bM R_n^\prime(\widehat\bTheta-\bTheta^\star)\\
    &-\sum_{g=1}^G  \lambda_g n^{1/2} \bT_g\bM \tilde \partial \|\bT_g \widehat\bTheta\|_2.
\end{aligned}
\end{equation}
From previous section, $\sqrt n (\bT_g\widehat\bTheta-\bT_g\bTheta^{(0)})=O_p(1)$, $\frac{l^\prime(\bTheta^\star)}{\sqrt n}=O_p(1)$, and $\frac{R_n^\prime(\widehat\bTheta-\bTheta^\star)}{\sqrt n}=o_p(1)$, and
\begin{equation*}
\lambda_g n^{1/2}=
    \begin{cases}
    O(\lambda n^{1/2}), &\quad \text{if } g\le g_0\\
    O(\lambda n^{(\gamma+1)/2}), &\quad \text{if } g> g_0
    \end{cases}
\end{equation*}
Also the sub-gradients are given as
\begin{equation*}
    \tilde \partial \|\bT_g \widehat\bTheta^\star\|_2=
    \begin{cases}
    \frac{\bT_g\trans \bT_g \widehat\bTheta}{\|\bT_g \widehat\bTheta\|_2}, &\quad \text{if } \|\bT_g \widehat\bTheta\|_2\ne0\\
    \bv, \; \text{ any }\bv:\|\bv\|_2\le 1,  &\quad \text{ if } \|\bT_g \widehat\bTheta\|_2=0.
    \end{cases}
\end{equation*}
Given that $\lambda n^{1/2} \to 0$ and $\lambda n^{(\gamma+1)/2}\to \infty$, if there are terms $\bT_g \widehat \bTheta \ne 0$ but $\bT_g \bTheta^\star=\0$, then the equation~\eqref{eq:sign_group} will be dominated by $\lambda_g n^{1/2} \bT_g\bM \tilde \partial \|\bT_g \widehat\bTheta\|_2$, which is of the order $O(\lambda n^{(\gamma+1)/2})$. Then the equation~\eqref{eq:sign_group} will not be true for sufficiently large $n$, since the other terms are of the order $O_p(1)$. Therefore, we could conclude that with probability tending to $1$, there is $\bT_g \widehat\bTheta=\0$ for any $\bT_g \bTheta^\star=\0$.
\end{proof}

\FloatBarrier

\clearpage

\section{Simulation Results}
\label{s:supp_sim}

\subsection{Additional Results for Baseline Settings}

The main simulation results in the paper focused on a baseline setting with $K=2$ clusters, $q=20$ covariates, and equal library sizes across samples. The complementary baseline results for heterogeneous-covariate selection and parameter estimation are reported in Supplement Tables~\ref{sim:tab:selection_bic_heterogeneous} and \ref{sim:tab:selection_bic_mse}. \revnew{Table~\ref{sim:tab:selection_bic_heterogeneous} summarizes how accurately the method identifies covariates whose effects differ across clusters, using sensitivity, specificity, and $F_1$ as classification summaries, whereas Table~\ref{sim:tab:selection_bic_mse} summarizes parameter-estimation accuracy under the baseline setting.
Smaller estimation errors indicate more accurate estimation, and DMMR-AP generally attains the smallest errors for $\bB$ and $\bDelta$.} \revnew{The reported errors are relative root-mean-squared errors: for a generic parameter $\boldeta$ with entries $\{\eta_j\}$ and estimate $\widehat\boldeta$, the rMSE is $\big(\sum_{j}(\widehat\eta_j-\eta_j)^2 \big/ \sum_{j}\eta_j^2\big)^{1/2}$, where the summation runs over all entries of the parameter (all components of $\bpi$ and $\btheta$, and all entries of $\bB$ and $\bDelta$).}

\begin{table}[htbp]
  \caption{Simulation: Heterogeneous covariates selection performance when $K=2$, $q=20$, and $M_i$ is fixed across samples.}
  \label{sim:tab:selection_bic_heterogeneous}
  \centering
  \begin{tabular}{llccc}
    \toprule
    $s$ & Model & Sensitivity & Specificity & $F_1$\\
    \midrule
    \multicolumn{5}{c}{$\theta = 0.05$} \\
    \multirow{2}{*}{0.2} & DMMR-P & 0.11 (0.30) & 1.00 (0.00) & 0.12 (0.31)\\
      & DMMR-AP & 1.00 (0.00) & 1.00 (0.00) & 1.00 (0.00)\\
    \multirow{2}{*}{0.4} & DMMR-P  & 1.00 (0.00) & 1.00 (0.00) & 1.00 (0.01)\\
      & DMMR-AP & 1.00 (0.00) & 1.00 (0.00) & 1.00 (0.00)\\
    \multirow{2}{*}{0.6} & DMMR-P & 0.96 (0.20) & 0.98 (0.08) & 0.94 (0.20)\\
      & DMMR-AP & 0.97 (0.15) & 1.00 (0.00) & 0.98 (0.13)\\
    \midrule
    \multicolumn{5}{c}{$\theta = 0.10$} \\
    \multirow{2}{*}{0.2} & DMMR-P & 0.00 (0.00) & 1.00 (0.00) & 0.00 (0.00)\\
      & DMMR-AP & 0.77 (0.41) & 1.00 (0.00) & 0.78 (0.41)\\
    \multirow{2}{*}{0.4} & DMMR-P  & 0.99 (0.07) & 1.00 (0.00) & 0.99 (0.07)\\
      & DMMR-AP & 0.99 (0.07) & 1.00 (0.00) & 0.99 (0.07)\\
    \multirow{2}{*}{0.6} & DMMR-P & 0.93 (0.26) & 0.99 (0.04) & 0.92 (0.26)\\
      & DMMR-AP & 0.97 (0.16) & 1.00 (0.02) & 0.97 (0.15)\\
    \bottomrule
  \end{tabular}
\end{table}

\begin{table}
  \caption{Simulation: Estimation performance when $K=2$, $q=20$, and $M_i$ is fixed across samples.}
  \vspace{-1em}
  \label{sim:tab:selection_bic_mse}
  \centering
  \begin{tabular}{llrrrr}
    \toprule
    $s$ & Model & $\mathrm{rMSE}(\bpi)$ & $\mathrm{rMSE}(\btheta)$ & $\mathrm{rMSE}(\bB)$ &$\mathrm{rMSE}(\bDelta)$ \\
    \midrule
    & \multicolumn{5}{c}{$\theta = 0.05$} \\
    \multirow{3}{*}{0.2} & DMMR-NP & 0.06 (0.04) & 0.25 (0.03) & 0.49 (0.03) & 0.49 (0.03)\\
     & DMMR-P & 0.06 (0.04) & 0.70 (0.20) & 0.48 (0.06) & 0.49 (0.05) \\
     & DMMR-AP & 0.06 (0.04) & 0.09 (0.04) & 0.23 (0.02) & 0.24 (0.02) \\
    \cmidrule(l{3pt}r{3pt}){1-6}
    \multirow{3}{*}{0.4} & DMMR-NP & 0.06 (0.04) & 0.29 (0.02) & 0.43 (0.03) & 0.43 (0.03)\\
     & DMMR-P & 0.06 (0.04) & 0.57 (0.21) & 0.36 (0.05) & 0.38 (0.06) \\
     & DMMR-AP & 0.06 (0.04) & 0.09 (0.04) & 0.18 (0.02) & 0.19 (0.02) \\
    \cmidrule(l{3pt}r{3pt}){1-6}
    \multirow{3}{*}{0.6} & DMMR-NP & 0.06 (0.05) & 0.41 (0.20) & 0.44 (0.10) & 0.45 (0.12)\\
     & DMMR-P & 0.06 (0.05) & 0.92 (0.87) & 0.36 (0.12) & 0.37 (0.13) \\
     & DMMR-AP & 0.06 (0.04) & 0.24 (0.72) & 0.19 (0.12) & 0.19 (0.13) \\
    \midrule
    & \multicolumn{5}{c}{$\theta = 0.10$} \\
    \multirow{3}{*}{0.2} & DMMR-NP & 0.06 (0.04) & 0.25 (0.03) & 0.65 (0.09) & 0.65 (0.09)\\
     & DMMR-P & 0.06 (0.04) & 0.57 (0.07) & 0.51 (0.03) & 0.52 (0.03) \\
     & DMMR-AP & 0.06 (0.04) & 0.20 (0.19) & 0.32 (0.09) & 0.33 (0.09) \\
    \cmidrule(l{3pt}r{3pt}){1-6}
    \multirow{3}{*}{0.4} & DMMR-NP & 0.06 (0.04) & 0.31 (0.03) & 0.60 (0.09) & 0.60 (0.10)\\
     & DMMR-P & 0.06 (0.05) & 0.57 (0.16) & 0.42 (0.06) & 0.44 (0.06) \\
     & DMMR-AP & 0.06 (0.04) & 0.12 (0.10) & 0.24 (0.04) & 0.24 (0.05) \\
    \cmidrule(l{3pt}r{3pt}){1-6}
    \multirow{3}{*}{0.6} & DMMR-NP & 0.06 (0.05) & 0.42 (0.08) & 0.62 (0.13) & 0.63 (0.15)\\
     & DMMR-P & 0.07 (0.07) & 0.96 (0.69) & 0.44 (0.14) & 0.45 (0.15) \\
     & DMMR-AP & 0.06 (0.05) & 0.26 (0.45) & 0.24 (0.11) & 0.24 (0.12) \\
    \bottomrule
  \end{tabular}
\end{table}

\clearpage
\subsection{Additional Simulation Results for Scalability and Robustness Evaluation}

To further evaluate the robustness of the proposed DMMR framework, we conducted additional simulations under varying read counts, higher-dimensional covariates, and larger numbers of clusters. The results were summarized in Tables~\ref{sim:tab:all_K_selection_add}, \ref{sim:tab:selection_bic_relevant_add}, \ref{sim:tab:selection_bic_heterogeneous_add}, and \ref{sim:tab:all_mse_add}, \revnew{reporting clustering accuracy, variable selection performance, and parameter estimation accuracy, respectively.}

As an additional simulation scenario, we investigated the impact of heterogeneous sequencing depths across samples. Although the main simulations fixed the total read count at $M=10,000$ to mirror the preprocessing of the STICS dataset \citep{zhou2019upper}, the proposed DMMR framework does not require equal library sizes. In our formulation, the sample-specific total count $M_i$ enters the Dirichlet---multinomial likelihood as a known quantity, allowing heterogeneous sequencing depths to be accommodated naturally. To assess robustness, we generated $M_i$ from $\{6000, 8000, 10000, 12000, 14000\}$ following \citet{weiss2017normalization}, while keeping all other components of the data-generating process unchanged. The results showed that clustering accuracy \revnew{(Table~\ref{sim:tab:all_K_selection_add})}, variable selection \revnew{(Tables~\ref{sim:tab:selection_bic_relevant_add} and \ref{sim:tab:selection_bic_heterogeneous_add})}, and parameter estimation \revnew{(Table~\ref{sim:tab:all_mse_add})} remained essentially unchanged relative to the fixed-$M$ setting, indicating that the proposed framework is robust to unequal library sizes.

We next examined performance under higher-dimensional covariates by considering scenarios with $q \in \{50,100\}$ while keeping the remaining simulation settings unchanged. The qualitative patterns remained consistent with the main simulations. DMMR-AP continued to achieve high sensitivity and $F_1$ scores in variable selection while maintaining high specificity. The same pattern was observed for coefficient estimation: DMMR-AP achieved smaller estimation errors for the regression coefficients $\bB$ and $\bDelta$ than both DMMR-P and the unpenalized estimator DMMR-NP. In contrast, the unpenalized estimator deteriorated substantially as the covariate dimension increased because the number of regression parameters grew rapidly. Notably, when $q$ was large, DMMR-NP often produced extremely small estimates of the dispersion parameter $\btheta$, leading to an rMSE close to $1$. This occurred because the large number of regression parameters allows the regression component to almost fully explain the cluster structure, effectively pushing the Dirichlet--multinomial model toward the multinomial case with negligible overdispersion. These results highlighted the stabilizing effect of adaptive penalization in high-dimensional settings.

Finally, we considered scenarios with a larger number of clusters ($K=5$). In this more complex setting, the advantage of penalization became more evident: both DMMR-P and DMMR-AP achieved substantially higher accuracy in selecting the correct number of clusters than the unpenalized estimator DMMR-NP\revnew{, as shown in Table~\ref{sim:tab:all_K_selection_add}}. Compared with the baseline setting ($K=2$), variable selection became more challenging because the number of cluster-specific regression parameters increased. As shown in Tables~\ref{sim:tab:selection_bic_relevant_add} and \ref{sim:tab:selection_bic_heterogeneous_add}, DMMR-P exhibited a noticeable decline in sensitivity and $F_1$ scores, whereas DMMR-AP maintained substantially better variable selection performance while preserving high specificity. Consistent with this pattern, DMMR-AP also achieved smaller estimation errors for $\bB$ and $\bDelta$, while DMMR-NP showed substantially larger errors due to the increased number of parameters when $K$ increased\revnew{, as shown in Table~\ref{sim:tab:all_mse_add}}. For the over-dispersion parameter $\btheta$, the penalized estimators generally exhibited larger estimation errors than DMMR-NP, likely reflecting a shrinkage---dispersion trade-off: shrinkage of regression coefficients may shift part of the unexplained variability toward the dispersion component, making accurate estimation of $\btheta$ more difficult.

\begin{table}[htbp]
  \centering
  \caption{Simulation: Accuracy of selecting $K$ (Acc($K$)) and the Kappa statistics (Kappa) for additional settings, fixing $\theta=0.05$.}
  \label{sim:tab:all_K_selection_add}
  \begin{tabular}{lrccc}
    \toprule
    & \multicolumn{2}{c}{$s = 0.4$} & \multicolumn{2}{c}{$s = 0.6$} \\
    \cmidrule(l{3pt}r{3pt}){2-3} \cmidrule(l{3pt}r{3pt}){4-5}
    & Acc($K$) & Kappa & Acc($K$) & Kappa\\
    \midrule
    & \multicolumn{4}{c}{$K=2,\, q = 50$} \\
    K-Means & 0.750 & 0.941 (0.058) & 0.075 & 0.889 (0.068)\\
    HC & 0.180 & 0.678 (0.250) & 0.020 & 0.108 (0.066)\\
    DMM & 0.885 & 0.884 (0.120) & 0.715 & 0.465 (0.239)\\
    DMMR-NP & 0.575 & 1.000 (0.002) & 0.920 & 1.000 (0.000)\\
    DMMR-P & 0.970 & 0.996 (0.049) & 0.810 & 0.986 (0.101)\\
    DMMR-AP & 0.920 & 1.000 (0.002) & 0.880 & 1.000 (0.000)\\
    \midrule
    & \multicolumn{4}{c}{$K=2,\, q = 100$} \\
    K-Means & 0.750 & 0.944 (0.045) & 0.070 & 0.880 (0.081)\\
    HC & 0.155 & 0.774 (0.189) & 0.015 & 0.111 (0.114)\\
    DMM & 0.875 & 0.882 (0.111) & 0.735 & 0.471 (0.250)\\
    DMMR-NP & 0.975 & 1.000 (0.001) & 0.925 & 0.946 (0.200)\\
    DMMR-P & 0.945 & 0.999 (0.004) & 0.760 & 0.966 (0.154)\\
    DMMR-AP & 0.695 & 1.000 (0.001) & 0.891 & 1.000 (0.000)\\
    \midrule
    & \multicolumn{4}{c}{$K=5,\, q = 20$} \\
    K-Means & 0.505 & 0.841 (0.059) & 0.215 & 0.396 (0.090)\\
    HC & 0.130 & 0.293 (0.092) & 0.075 & 0.195 (0.055)\\
    DMM & 0.005 & - & 0.000 & - \\
    DMMR-NP & 0.035 & 0.996 (0.007) & 0.080 & 0.966 (0.041)\\
    DMMR-P & 0.495 & 0.882 (0.075) & 0.260 & 0.479 (0.117)\\
    DMMR-AP & 0.635 & 0.967 (0.068) & 0.350 & 0.689 (0.266)\\
    \midrule
    & \multicolumn{4}{c}{$K=2,\, q = 20$, $M_i$ varying across samples} \\
    K-Means & 0.755 & 0.944 (0.049) & 0.070 & 0.904 (0.080)\\
    HC & 0.095 & 0.768 (0.184) & 0.000 & -\\
    DMM & 0.935 & 0.886 (0.116) & 0.760 & 0.485 (0.249)\\
    DMMR-NP & 1.000 & 1.000 (0.001) & 0.970 & 1.000 (0.001)\\
    DMMR-P & 0.980 & 1.000 (0.001) & 0.920 & 1.000 (0.002)\\
    DMMR-AP & 1.000 & 1.000 (0.001) & 0.975 & 0.995 (0.061)\\
    \bottomrule
  \end{tabular}
\end{table}

\begin{table}[htbp]
  \caption{Simulation: Relevant covariates selection performance for additional settings, fixing $\theta=0.05$.}
  \label{sim:tab:selection_bic_relevant_add}
  \centering
  \begin{tabular}{llccc}
    \toprule
    $s$ & Model & Sensitivity & Specificity & $F_1$\\
    \midrule
    \multicolumn{5}{c}{$K=2,\, q = 50$} \\
    0.4 & DMMR-P & 1.00 (0.04) & 0.96 (0.04) & 0.93 (0.06)\\
    0.4 & DMMR-AP & 1.00 (0.00) & 0.98 (0.02) & 0.97 (0.04)\\
    0.6 & DMMR-P & 0.93 (0.18) & 0.91 (0.12) & 0.83 (0.14)\\
    0.6 & DMMR-AP & 0.96 (0.13) & 0.98 (0.05) & 0.95 (0.10)\\
    \midrule
    \multicolumn{5}{c}{$K=2,\, q = 100$} \\
    0.4 & DMMR-P & 1.00 (0.00) & 0.98 (0.02) & 0.92 (0.07)\\
    0.4 & DMMR-AP & 1.00 (0.00) & 0.99 (0.01) & 0.96 (0.04)\\
    0.6 & DMMR-P & 0.88 (0.21) & 0.97 (0.08) & 0.82 (0.13)\\
    0.6 & DMMR-AP & 0.96 (0.14) & 0.95 (0.21) & 0.92 (0.19)\\
    \midrule
    \multicolumn{5}{c}{$K=5,\, q = 20$} \\
    0.4 & DMMR-P & 0.53 (0.06) & 1.00 (0.02) & 0.69 (0.05)\\
    0.4 & DMMR-AP & 0.78 (0.25) & 0.93 (0.10) & 0.82 (0.14)\\
    0.6 & DMMR-P & 0.52 (0.05) & 0.99 (0.03) & 0.68 (0.04)\\
    0.6 & DMMR-AP & 0.62 (0.21) & 0.96 (0.10) & 0.73 (0.12)\\
    \midrule
    \multicolumn{5}{c}{$K=2,\, q = 20$, $M_i$ varying across samples} \\
    0.4 & DMMR-P & 1.00 (0.00) & 0.77 (0.26) & 0.91 (0.09)\\
    0.4 & DMMR-AP & 1.00 (0.00) & 1.00 (0.02) & 1.00 (0.01)\\
    0.6 & DMMR-P & 0.99 (0.08) & 0.42 (0.39) & 0.79 (0.12)\\
    0.6 & DMMR-AP & 0.99 (0.06) & 0.99 (0.07) & 0.99 (0.05)\\
    \bottomrule
  \end{tabular}
\end{table}

\begin{table}[htbp]
  \caption{Simulation: Heterogeneous covariates selection performance for additional settings, fixing $\theta=0.05$.}
  \label{sim:tab:selection_bic_heterogeneous_add}
  \centering
  \begin{tabular}{llccc}
    \toprule
    $s$ & Model & Sensitivity & Specificity & $F_1$\\
    \midrule
    \multicolumn{5}{c}{$K=2,\, q = 50$} \\
    0.4 & DMMR-P & 0.99 (0.07) & 1.00 (0.00) & 0.99 (0.07)\\
    0.4 & DMMR-AP & 1.00 (0.00) & 1.00 (0.00) & 1.00 (0.00)\\
    0.6 & DMMR-P & 0.85 (0.36) & 1.00 (0.00) & 0.85 (0.36)\\
    0.6 & DMMR-AP & 0.92 (0.27) & 1.00 (0.00) & 0.92 (0.26)\\
    \midrule
    \multicolumn{5}{c}{$K=2,\, q = 100$} \\
    0.4 & DMMR-P & 1.00 (0.00) & 1.00 (0.00) & 1.00 (0.00)\\
    0.4 & DMMR-AP & 1.00 (0.00) & 1.00 (0.00) & 1.00 (0.00)\\
    0.6 & DMMR-P & 0.77 (0.43) & 0.99 (0.07) & 0.76 (0.43)\\
    0.6 & DMMR-AP & 0.91 (0.28) & 0.96 (0.20) & 0.87 (0.32)\\
    \midrule
    \multicolumn{5}{c}{$K=5,\, q = 20$} \\
    0.4 & DMMR-P & 0.00 (0.00) & 1.00 (0.00) & 0.00 (0.00)\\
    0.4 & DMMR-AP & 0.56 (0.50) & 1.00 (0.01) & 0.55 (0.49)\\
    0.6 & DMMR-P & 0.00 (0.00) & 1.00 (0.00) & 0.00 (0.00)\\
    0.6 & DMMR-AP & 0.23 (0.42) & 1.00 (0.01) & 0.23 (0.42)\\
    \midrule
    \multicolumn{5}{c}{$K=2,\, q = 20$, $M_i$ varying across samples} \\
    0.4 & DMMR-P & 1.00 (0.00) & 1.00 (0.01) & 1.00 (0.01)\\
    0.4 & DMMR-AP & 1.00 (0.00) & 1.00 (0.00) & 1.00 (0.00)\\
    0.6 & DMMR-P & 0.97 (0.17) & 0.97 (0.12) & 0.95 (0.18)\\
    0.6 & DMMR-AP & 0.98 (0.12) & 1.00 (0.00) & 0.98 (0.11)\\
    \bottomrule
  \end{tabular}
\end{table}

\begin{table}[htbp]
  \caption{Simulation: Estimation performance for additional settings, fixing $\theta=0.05$.}
  \label{sim:tab:all_mse_add}
  \centering
  \begin{tabular}{llrrrr}
    \toprule
    $s$ & Model & $\mathrm{rMSE}(\bpi)$ & $\mathrm{rMSE}(\btheta)$ & $\mathrm{rMSE}(\bB)$ &$\mathrm{rMSE}(\bDelta)$ \\
    \midrule
    & \multicolumn{5}{c}{$K=2,\, q = 50$} \\
    \multirow{3}{*}{0.4} & DMMR-NP & 0.06 (0.04) & 0.73 (0.03) & 1.88 (0.30) & 1.88 (0.30)\\
     & DMMR-P & 0.06 (0.04) & 0.85 (0.15) & 0.43 (0.03) & 0.45 (0.04) \\
     & DMMR-AP & 0.06 (0.04) & 0.28 (0.09) & 0.26 (0.03) & 0.27 (0.03) \\
     \cmidrule(l{3pt}r{3pt}){1-6}
    \multirow{3}{*}{0.6} & DMMR-NP & 0.06 (0.08) & 0.82 (0.14) & 2.19 (0.25) & 2.19 (0.25)\\
     & DMMR-P & 0.09 (0.13) & 2.09 (1.22) & 0.52 (0.14) & 0.55 (0.16) \\
     & DMMR-AP & 0.07 (0.09) & 0.75 (1.10) & 0.30 (0.16) & 0.31 (0.17) \\
    \midrule
    & \multicolumn{5}{c}{$K=2,\, q = 100$} \\
    \multirow{3}{*}{0.4} & DMMR-NP & 0.06 (0.04) & 1.00 (0.00) & 4.74 (0.75) & 4.64 (0.74)\\
     & DMMR-P & 0.06 (0.04) & 0.97 (0.15) & 0.45 (0.03) & 0.48 (0.03) \\
     & DMMR-AP & 0.06 (0.04) & 0.41 (0.10) & 0.30 (0.03) & 0.32 (0.03) \\
     \cmidrule(l{3pt}r{3pt}){1-6}
    \multirow{3}{*}{0.6} & DMMR-NP & 0.07 (0.08) & 1.00 (0.02) & 2.59 (0.86) & 2.49 (0.84)\\
     & DMMR-P & 0.10 (0.12) & 2.56 (1.32) & 0.58 (0.16) & 0.61 (0.17) \\
     & DMMR-AP & 0.07 (0.09) & 0.73 (1.13) & 0.36 (0.36) & 0.37 (0.35) \\
    \midrule
    & \multicolumn{5}{c}{$K=5,\, q = 20$} \\
    \multirow{3}{*}{0.4} & DMMR-NP & 0.19 (0.14) & 0.67 (0.15) & 1.82 (0.21) & 1.96 (0.22)\\
     & DMMR-P & 0.23 (0.14) & 2.13 (0.19) & 0.69 (0.09) & 0.74 (0.10) \\
     & DMMR-AP & 0.20 (0.14) & 1.05 (0.88) & 0.51 (0.20) & 0.55 (0.21) \\
     \cmidrule(l{3pt}r{3pt}){1-6}
    \multirow{3}{*}{0.6} & DMMR-NP & 0.41 (0.15) & 0.46 (0.24) & 1.62 (0.22) & 1.80 (0.25)\\
     & DMMR-P & 0.53 (0.17) & 4.57 (0.35) & 0.90 (0.08) & 0.99 (0.09) \\
     & DMMR-AP & 0.47 (0.21) & 3.35 (1.42) & 0.79 (0.20) & 0.88 (0.23) \\
    \midrule
    & \multicolumn{5}{c}{$K=2,\, q = 20$, $M_i$ varying across samples} \\
    \multirow{3}{*}{0.4} & DMMR-NP & 0.06 (0.04) & 0.29 (0.03) & 0.44 (0.04) & 0.44 (0.04)\\
     & DMMR-P & 0.06 (0.04) & 0.59 (0.20) & 0.36 (0.05) & 0.38 (0.06) \\
     & DMMR-AP & 0.06 (0.04) & 0.09 (0.05) & 0.19 (0.02) & 0.19 (0.02) \\
     \cmidrule(l{3pt}r{3pt}){1-6}
    \multirow{3}{*}{0.6} & DMMR-NP & 0.06 (0.08) & 0.42 (0.35) & 0.44 (0.10) & 0.45 (0.11)\\
     & DMMR-P & 0.07 (0.08) & 0.82 (0.85) & 0.34 (0.12) & 0.35 (0.13) \\
     & DMMR-AP & 0.06 (0.08) & 0.20 (0.63) & 0.18 (0.11) & 0.18 (0.12) \\
    \bottomrule
  \end{tabular}
\end{table}

Table~\ref{sim:time_cost} summarizes the computational time for representative settings.
We set the sample size to $n=200$ and the number of taxa to $p=20$, and vary the number of covariates $q$ and clusters $K$ to evaluate computational scalability under increasing model complexity. For each configuration, the reported time corresponds to fitting the model at candidate $\hat K$ values. The computational cost increases with $q$ and
$\hat K$.

 \begin{table}[htbp]
      \caption{Simulation: Computational time for representative settings ($\theta=0.05, s=0.4$) in minutes.}
      \label{sim:time_cost}
      \centering
      \setlength{\tabcolsep}{8pt}
      \begin{tabular}{lcrr rr}
        \toprule
        & & \multicolumn{2}{c}{DMMR-P} & \multicolumn{2}{c}{DMMR-AP} \\
        \cmidrule(l{3pt}r{3pt}){3-4} \cmidrule(l{3pt}r{3pt}){5-6}
$(K,q)$ & $\hat K$ & ADMM & Overall & ADMM & Overall\\
        \midrule
        \multirow{3}{*}{(2, 20)}
 & 1 & 0.2 (0.0) & 0.4 (0.1) & 0.3 (0.1) & 0.6 (0.1)\\
 & 2 & 1.8 (0.4) & 3.2 (0.5) & 1.8 (0.4) & 3.0 (0.5)\\
 & 3 & 4.7 (1.3) & 7.5 (2.5) & 3.6 (1.1) & 5.7 (1.4)\\
        \midrule
        \multirow{3}{*}{(2, 50)}
 & 1 & 0.3 (0.1) & 0.6 (0.1) & 0.6 (0.1) & 0.9 (0.1)\\
 & 2 & 5.3 (2.4) & 7.4 (2.8) & 4.4 (2.0) & 6.3 (2.4)\\
 & 3 & 8.3 (3.5) & 11.6 (4.3) & 5.6 (2.6) & 8.2 (3.1)\\
        \midrule
        \multirow{3}{*}{(2, 100)}
 & 1 & 0.8 (0.1) & 1.2 (0.2) & 1.7 (0.3) & 2.1 (0.3)\\
 & 2 & 17.9 (9.4) & 21.7 (10.2) & 16.3 (8.1) & 19.3 (8.9)\\
 & 3 & 21.3 (14.0) & 25.7 (15.2) & 10.7 (7.9) & 13.8 (8.6)\\
        \midrule
        \multirow{3}{*}{(5, 20)}
 & 3 & 3.6 (0.6) & 6.3 (0.9) & 3.2 (0.6) & 5.5 (0.8)\\
 & 4 & 5.8 (1.0) & 9.7 (1.4) & 5.0 (1.0) & 8.4 (1.3)\\
 & 5 & 9.7 (2.6) & 15.3 (3.4) & 7.5 (1.9) & 12.1 (2.6)\\
 & 6 & 16.4 (3.8) & 24.6 (5.4) & 10.5 (3.1) & 16.3 (4.1)\\
\bottomrule
\end{tabular}
\end{table}

\FloatBarrier

\clearpage
\section{Application}
\label{s:supp_application}

\subsection{Covariate Description}\label{supp:app:covariate}

\begin{table}[htp]
\caption{Selected demographic and clinical features from the STICS dataset. The table includes 14 original variables, expanded into 18 dummy variables. The first four variables are continuous, while the remaining are binary or categorical.}
\label{tab:app:covariates}
\centering
\scriptsize
\begin{tabular}{@{}lll@{}}
\toprule
Variable Name & Description & Coding \\
\midrule
age\_enr & Age at enrollment & Continuous \\
bmi & Body Mass Index (BMI) & Continuous \\
num\_oral\_steroid\_courses & Number of oral steroid courses for asthma in the past year & Continuous \\
ige & Immunoglobulin E (IgE) level & Continuous \\
gender & Gender & 1 = Male, 2 = Female \\
race & Race & White, Black, Other \\
ethnicity & Ethnicity & 1 = Non-Hispanic, 2 = Hispanic \\
parent\_ast & Parental history of asthma & 0 = No, 1 = Yes, 8 = Don't know \\
smoke\_exp & Tobacco smoke exposure & 0 = No, 1 = Yes, 8 = Don't know \\
pets & Pet exposure & 0 = No, 1 = Yes \\
eczema & Participant history of eczema & 0 = No, 1 = Yes, 8 = Don't know \\
steroid & Nasal steroid use prior to enrollment & 0 = No, 1 = Yes \\
antibiotics & Antibiotic use prior to enrollment & 0 = No, 1 = Yes \\
virus & Viral analysis result at baseline & 0 = Negative, 1 = Positive \\
\bottomrule
\end{tabular}
\end{table}

\subsection{Additional Estimation Results}\label{supp:app:para}

Table~\ref{app:tab:pi_theta} summarizes the estimated mixing proportions ($\bpi$) and over-dispersion parameters ($\btheta$) for HC+DM and DMMR-AP.

\begin{table}[htp]
\centering
\caption{Application: Estimates of parameters $\bpi$ and $\btheta$.}
\vspace{-1em}
\label{app:tab:pi_theta}
\begin{tabular}{lrrrr}
\toprule
\multicolumn{1}{c}{ } & \multicolumn{2}{c}{HC+DM} & \multicolumn{2}{c}{DMMR-AP} \\
\cmidrule(l{3pt}r{3pt}){2-3} \cmidrule(l{3pt}r{3pt}){4-5}
Cluster & $\bpi$ & $\btheta$ & $\bpi$ & $\btheta$\\
\midrule
 Strep-dominant & 0.2897 & 0.0840 & 0.5159 & 0.1053\\
 Dolo/Coryne-dominant & 0.5374 & 0.1654 & 0.3346 & 0.1053\\
 Mixed-pathobiont & 0.1729 & 0.1896 & 0.1495 & 0.1053\\
\bottomrule
\vspace{1em}
\end{tabular}
\end{table}

\subsection{Sensitivity to the Taxa Abundance Threshold}\label{supp:app:abund_threshold}

We assessed sensitivity to the abundance threshold by repeating the analysis with a lower cutoff of $0.1\%$, which increased the number of retained taxa from $p = 24$ to $p = 41$. We then refit DMMR under the expanded taxa set using the same model-selection procedure as in the main analysis.

The clustering solution was stable across the two settings of the abundance threshold: the Adjusted Rand Index (ARI) of DMMR clustering patterns was $0.82$. For comparison, the corresponding ARI for HC was $0.74$, indicating slightly weaker agreement under the threshold change.

The covariate findings were also largely preserved. Under $0.5\%$ threshold and $p = 24$, ``eczema1'' and ``steroid1'' had heterogeneous effects across the three clusters. Under $0.1\%$ threshold and $p = 41$, ``eczema1'' remained selected, while ``viruspositive'' replaced ``steroid1''. Thus, the primary signal persisted and only the secondary heterogeneous covariate changed.

Overall, these results indicated that the proposed DMMR was reasonably stable with respect to the abundance threshold, with clustering structure and primary covariate effects largely preserved. We would like to point out that the choice of threshold involves a trade-off between including more taxa (which may capture additional biological variations) and maintaining model stability/data quality (which may enhance interpretability and reduce noise). We thus recommend that researchers consider the specific context and quality of their data when selecting an appropriate cutoff.

\begin{figure}[htbp]
    \centering
    \includegraphics[width=\linewidth]{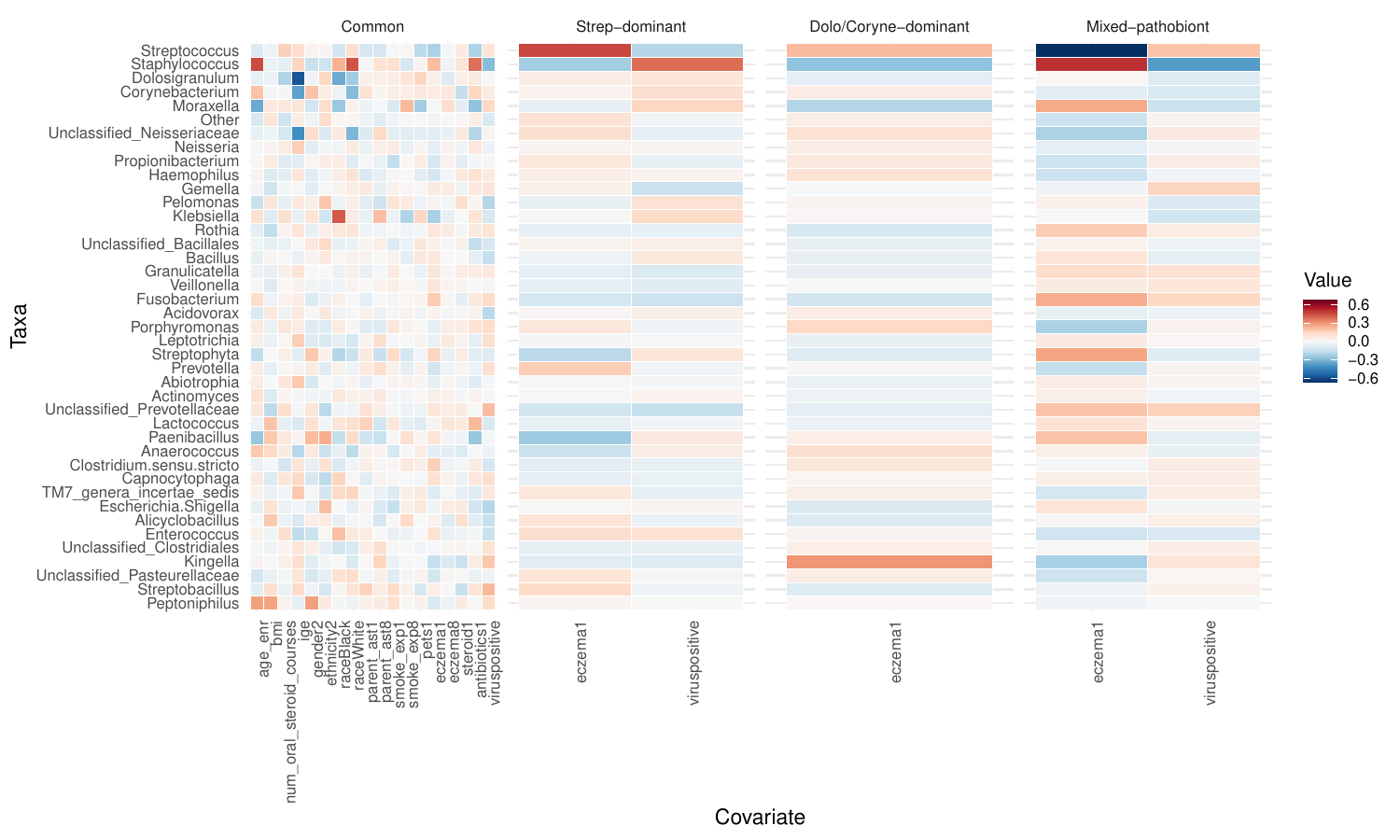}
    \caption{Application: Estimated covariate effects for the three clusters identified by DMMR with the abundance threshold set to 0.1\%. The common effects of all covariates and the cluster-specific effects of two selected heterogeneous covariates are shown.}

    \label{app:fig:coef_41taxa}
  \end{figure}

\subsection{Survival Analysis of Time to YZ Episode}\label{supp:app:survival}

We then compared the time-to-event distributions of YZ across the three clusters identified by both models. To facilitate the comparison, we closely followed the setup in \citet{zhou2019upper}.
Participants were followed for up to 320 days. The event of interest was defined as developing at least 2 YZ episodes during the follow-up. Among the 214 participants, 93 experienced two or more YZ episodes, while the remaining 121 either had none or only one. \rev{The Kaplan---Meier curves corresponding to different clustering approaches were shown in Figure~\ref{app:fig:survival}. Under HC, some visual separation emerged after approximately 200 days, with the \textit{Dolo/Coryne-dominant} cluster exhibiting higher survival probabilities than the other two clusters. However, the overall log-rank test was not statistically significant at the 0.05 level ($p = 0.11$). In contrast, the survival curves based on the DMMR clusters showed noticeably different separation patterns and approached statistical significance at the 0.05 level ($p = 0.053$). In particular, the DMMR curves exhibited consistent separation of survival probabilities between approximately 100 and 300 days, with the \textit{Dolo/Coryne-dominant} cluster maintaining the highest survival probabilities and the \textit{Strep-dominant} cluster consistently the lowest; after 300 days, the \textit{Strep-dominant} and \textit{Mixed-pathobiont} clusters converged, while the \textit{Dolo/Coryne-dominant} cluster remained higher.}

\begin{figure}[htp]
\centering
    \begin{subfigure}[t]{0.8\linewidth}
    \centering
    \includegraphics[width=\linewidth]{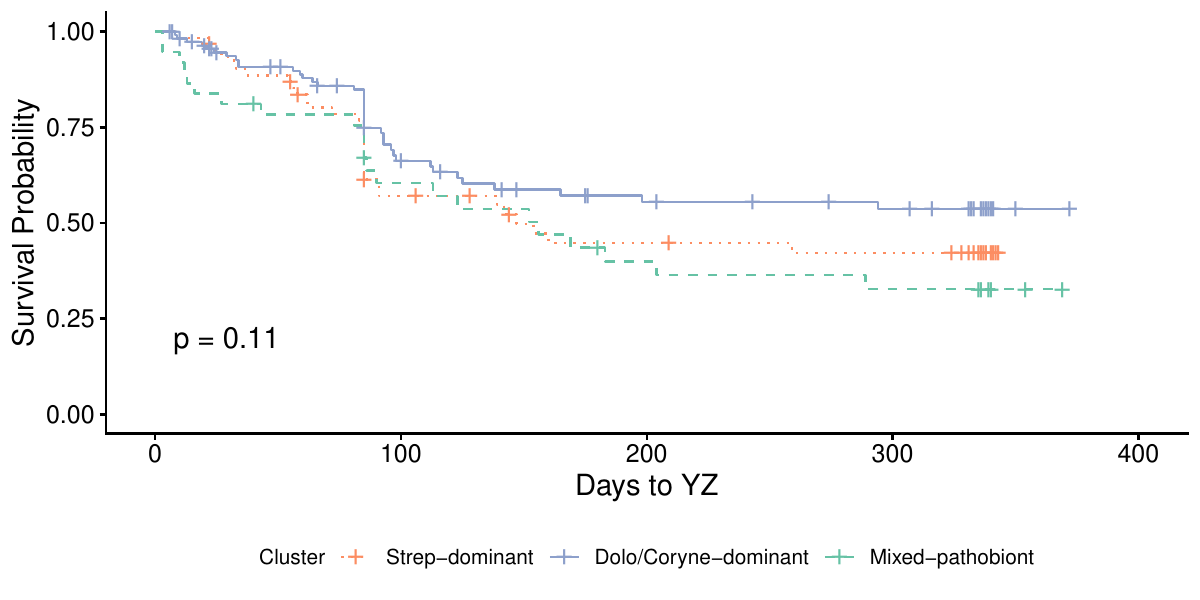}
    \caption{HC}
    \end{subfigure}\hspace{0.01\textwidth}
    \begin{subfigure}[t]{0.8\linewidth}
    \centering
    \includegraphics[width=\linewidth]{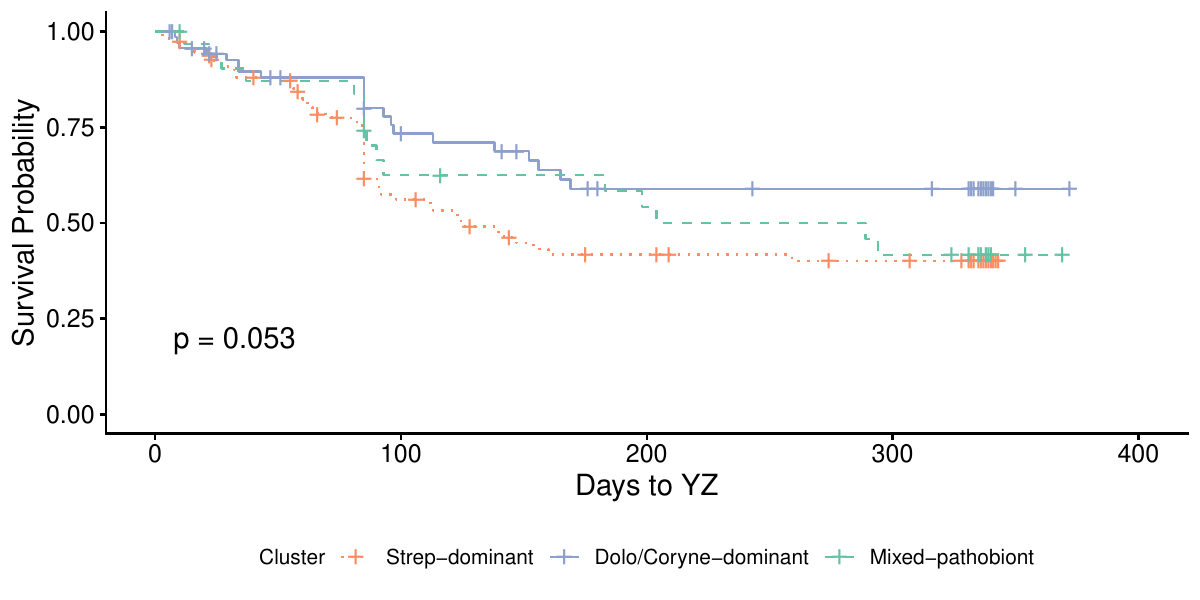}
    \caption{DMMR-AP}
    \end{subfigure}\hspace{0.01\textwidth}
    \caption{Application: Kaplan-Meier curves of developing YZ for the clusters identified by either HC or DMMR.}
    \label{app:fig:survival}
\end{figure}

\FloatBarrier

\end{document}